\documentclass[sigplan,10pt]{acmart}

\copyrightyear{2024} 
\acmYear{2024} 
\setcopyright{acmlicensed}\acmConference[EuroSys '24]{Nineteenth European Conference on Computer Systems}{April 22--25, 2024}{Athens, Greece}

\acmBooktitle{Nineteenth European Conference on Computer Systems (EuroSys '24), April 22--25, 2024, Athens, Greece}
\acmPrice{15.00}
\acmDOI{10.1145/3627703.3629559}
\acmISBN{979-8-4007-0437-6/24/04}

\usepackage{lipsum}

\usepackage{tikz}
\usepackage{amsmath}

\newcommand{\system}{Dordis}

\usepackage{xurl}

\usepackage{cleveref}
\crefformat{section}{§#2#1#3}

\usepackage{subcaption}

\newcommand{\PHB}[1]{\noindent\textbf{#1}\hspace{.5em}} 
\newcommand{\PHM}[1]{\vspace{.2em}\noindent\textbf{#1}\hspace{.5em}} 

\usepackage{amsthm}
\usepackage{amsfonts}
\usepackage[thinc]{esdiff}
\newtheorem{lemma}{Lemma}
\newtheorem{theorem}{Theorem}
\newtheorem{definition}{Definition}

\usepackage{booktabs}
\usepackage{tabularx}
\usepackage{multirow}
\usepackage{makecell}

\usepackage{enumitem}
\setlist{leftmargin=3mm, itemsep=0mm}

\usepackage{bbm}

\usepackage{bm}

\newtheoremstyle{case}{}{}{}{}{}{:}{ }{}
\theoremstyle{case}

\DeclareMathOperator*{\argmin}{arg\,min}

\definecolor{myred}{rgb}{0.58, 0.06, 0}

\newcommand{\xnoise}[1]{\textcolor{myred}{\underline{#1}}}

\usepackage{soul}
\definecolor{mygreen}{rgb}{0.88, 0.91, 0.89}
\sethlcolor{mygreen}
\newcommand{\reuse}[1]{\hl{#1}}

\newcommand{\malicious}[1]{\textit{[#1]}}

\makeatletter
\newtheorem*{rep@theorem}{\rep@title}
\newcommand{\newreptheorem}[2]{%
	\newenvironment{rep#1}[1]{%
		\def\rep@title{#2 \ref{##1}}%
		\begin{rep@theorem}}%
		{\end{rep@theorem}}}
\makeatother
\newreptheorem{theorem}{\textbf{Theorem}}
\newreptheorem{lemma}{Lemma}

\makeatletter
\def\@ACM@checkaffil{
    \if@ACM@instpresent\else
    \ClassWarningNoLine{\@classname}{No institution present for an affiliation}%
    \fi
    \if@ACM@citypresent\else
    \ClassWarningNoLine{\@classname}{No city present for an affiliation}%
    \fi
    \if@ACM@countrypresent\else
        \ClassWarningNoLine{\@classname}{No country present for an affiliation}%
    \fi
}
\makeatother

\begin{document}

\title{\system{}: Efficient Federated Learning with Dropout-Resilient Differential Privacy}

\author{Zhifeng Jiang}
\affiliation{%
  \institution{HKUST}
}
\email{zjiangaj@cse.ust.hk}

\author{Wei Wang}
\affiliation{%
  \institution{HKUST}
}
\email{weiwa@cse.ust.hk}

\author{Ruichuan Chen}
\affiliation{%
  \institution{Nokia Bell Labs}
}
\email{ruichuan.chen@gmail.com}


\begin{abstract}
Federated learning (FL) is increasingly deployed among multiple clients to train a shared model over decentralized data.
To address privacy concerns, FL systems need to safeguard the clients' data from disclosure during training and control data leakage through trained models when exposed to untrusted domains.
Distributed differential privacy (DP) offers an appealing solution in this regard as it achieves a balanced tradeoff between privacy and utility without a trusted server.
However, existing distributed DP mechanisms are impractical in
the presence of \emph{client dropout}, resulting in poor privacy guarantees or degraded training accuracy.
In addition, these mechanisms suffer from severe efficiency issues.

We present \system{}, a distributed differentially private
FL framework that is highly efficient and resilient to client dropout.
Specifically, we develop a novel `add-then-remove' scheme that enforces a required noise level precisely in each training round, even if some sampled clients drop out.
This ensures that the privacy budget is utilized prudently, despite unpredictable client dynamics.
To boost performance, \system{} operates as a distributed parallel architecture via encapsulating the communication and computation operations into stages.
It automatically divides the global model aggregation into several chunk-aggregation tasks and pipelines them for optimal speedup.
Large-scale deployment evaluations demonstrate that \system{} efficiently handles client dropout in various realistic FL scenarios, achieving the optimal privacy-utility tradeoff and accelerating training by up to 2.4$\times$ compared to existing solutions.
\end{abstract}

\ccsdesc[500]{Computing methodologies~Supervised learning}
\ccsdesc[300]{Security and privacy~Database and storage security}

\keywords{Federated Learning, Distributed Differential Privacy, Client Dropout, Secure Aggregation, Pipeline}

\received{24 May 2023}
\received[accepted]{12 September 2023}

\maketitle

\section{Introduction}
\label{sec:introduction}

Federated learning (FL)~\cite{mcmahan2017communication,kairouz2019advances} enables collaborative training of a shared model among multiple clients (e.g., mobile and edge devices) under the orchestration of a central server.
In scenarios where the number of clients is large (e.g., millions of mobile devices), the server dynamically \emph{samples} a small subset of clients to participate in each training round~\cite{bonawitz2019towards}.
These clients download the global model from the server, compute local updates using private data, and upload these updates to the server for global aggregation.
Throughout the training, no client's data is exposed directly.
FL has been deployed in various domains, enabling a multitude of privacy-sensitive AI
applications~\cite{google2021assistant, paulik2021federated, ludwig2020ibm,
webank2020laundering, li2019privacy, nvidia2020oxygen}.

However, solely keeping client data on the device is inadequate for preserving data privacy.
Recent work has shown that sensitive data can still be exposed through message exchanges in  FL training~\cite{geiping2020inverting, yin2021see, wang2022protect, zhao2023secure}. 
It is also possible to infer a client's data from the trained models~\cite{carlini2019secret, song2019auditing, shokri2017membership, nasr2019comprehensive} by exploiting their ability to memorize information~\cite{carlini2019secret, song2019auditing}.

Current FL systems often use \emph{differential privacy} (DP)~\cite{dwork2014algorithmic} to perturb the aggregate model update in each round, limiting the disclosure of individual clients' data throughout training.
Among the three typical DP models (i.e., central~\cite{mcmahan2018learning, ramaswamy2020training}, local~\cite{pihur2018differentially}, and distributed DP~\cite{kairouz2021distributed, agarwal2021skellam, stevens2022efficient}), \emph{distributed DP} is the most appealing for FL as it: 1) assumes \emph{no trusted server} (in contrast to central DP), and 2) imposes \emph{minimum noise} given a privacy budget (unlike local DP), causing little utility loss to the trained models.
Specifically, given a global privacy budget that must not be exceeded, the system first calculates the minimum random noise required in each round.
Then, in each round, every sampled client adds a small portion of the required noise to its local update.
The aggregate update at the server is thus perturbed by \emph{exactly} the minimum required noise.
In distributed DP, local updates are aggregated using secure aggregation, which ensures that the (untrusted) server learns only the aggregate result, not individual updates.

Existing distributed DP mechanisms, however, face two practical challenges when deployed in real-world FL systems.
First, the privacy guarantee of distributed DP can be compromised in the presence of \emph{client dropout}, which can occur at any time due to network errors, low battery, or changes in eligibility, as frequently observed in production~\cite{bonawitz2019towards, yang2021characterizing, lai2022fedscale}.
For dropped clients, their noise contributions are missing in the aggregate update, leading to insufficient DP protection and consuming more privacy budget than initially allocated in each round~\cite{kairouz2021distributed}.
This exhausts the privacy budget quickly, resulting in early termination of training and significant model utility loss (\cref{sec:issue_client_dropout}).
As a quick solution, proactively increasing the amount of noise used may not achieve the optimal privacy-utility tradeoff without strong expertise in client dynamics and learning tasks.

Second, aside from the privacy issue caused by client dropout, the secure aggregation protocols (e.g., SecAgg \cite{bonawitz2017practical}) employed in distributed DP present severe efficiency challenges.
This is due to the multi-round communications and heavy cryptographic computations involved, which can consume up to 97\% of the time of each training round, as observed in real deployments (\cref{sec:secagg_issue}).
Recent attempts~\cite{bell2020secure, so2021turbo, kadhe2020fastsecagg, so2022lightsecagg} to improve the asymptotic complexity sacrifice desired properties (e.g., malicious security or dropout tolerance), and/or have limited practical gains.

In this paper, we present \system{}, an efficient FL framework that enables \underline{\textbf{d}}rop\underline{\textbf{o}}ut-\underline{\textbf{r}}esilient \underline{\textbf{dis}}tributed DP in FL.
\system{} addresses both the privacy and efficiency issues of distributed DP in FL systems with two key contributions.
First, to ensure that the aggregate model update is perturbed with exactly the minimum required noise regardless of client dropout, we
devise a novel \emph{add-then-remove} scheme named \texttt{XNoise} (\cref{sec:enforcement}).
\texttt{XNoise} initially lets each selected client to add excessive noise to its local update.
After aggregation, the server removes part of the excessive noise that exceeds the minimum noise requirement, based on the actual client participation. 
To cope with the potential failures in executing \texttt{XNoise}, we consolidate its security using efficient cryptographic primitives.
\texttt{XNoise} is proven to preserve the privacy of honest clients, even in the presence of a \emph{malicious} server colluding with a small subset of other clients.

Second, to expedite the execution of distributed DP, we run \system{} as a \emph{distributed parallel architecture} to overlap computation- and communication-intensive operations (e.g., data encoding and transmission) (\cref{sec:pipeline}).
To enable a generic design, \system{} first abstracts the distributed DP workflow into a sequence of \emph{stages} with different dominant system resources.
By dividing the global aggregation task into several \emph{chunk}-aggregation task, \system{} allows pipeline parallelism by scheduling them to run different stages concurrently.
With a realistic performance model and profiling technique, \system{} can identify the optimal pipeline configuration by solving an optimization problem for \emph{maximum speedup}.

We implemented \system{} as an end-to-end system with generic designs that support a wide range of distributed DP protocols (\cref{sec:implementation}).\footnote{\system{} is available at \url{https://github.com/SamuelGong/Dordis}.}
We deployed \system{} in a real distributed environment that
features data and hardware heterogeneity of client devices (\cref
{sec:evaluation}).
Our evaluations across various FL training tasks show that the necessary noise for aggregated updates in \system{} can be enforced without impairing model utility in the presence of client dropout.
Moreover, \system{}'s pipeline execution speeds up the baseline systems with different secure aggregation protocols by up to 2.4$\times$, without reducing their security properties.

\section{Background and Motivation}
\label{sec:background}

\subsection{Scenario}
\label{sec:background_threat}

\PHB{Federated Learning.} Federated learning (FL) enables a large number
 of clients (e.g., millions of mobile and edge devices) to
 collaboratively build a global model without sharing local data~\cite
 {mcmahan2017communication,kairouz2019advances}. In FL, a (logically) centralized server maintains the global model and orchestrates the
 iterative training. At the beginning of each training round, the
 server randomly samples a subset of available clients as
 participants~\cite{bonawitz2019towards}. The sampled clients perform local
 training to the downloaded global model using their private data and report only
 the model updates to the server. The server collects the updates from participants until a certain
 deadline, and aggregates these updates. It then uses the aggregate update (e.g., FedAvg~\cite
 {mcmahan2017communication}) to refine the global model.

\PHM{Threat Model.} 
Although client data is not directly exposed in the FL process,
a large body of research has shown that it is still possible to reveal sensitive client data
from individual updates or trained models via \emph
{data reconstruction} or \emph{membership inference} attacks.
For example, an adversary can accurately reconstruct a client's training data from its gradient updates~\cite{geiping2020inverting, yin2021see, wang2022protect, zhao2023secure}; an adversary can also infer from a
trained language model whether a client's private text is in the training
set~\cite{carlini2019secret, song2019auditing, shokri2017membership, nasr2019comprehensive}. 

We aim to control the exposure of honest clients' data against the above-mentioned attacks. Following the Secure Multi-Party Computation (SMPC) literature~\cite{evans2018pragmatic}, we target both the \emph{semi-honest} setting (where all parties faithfully follow the protocol) and the \emph{malicious} setting (where the adversary can deviate arbitrarily from the protocol). In both settings, the adversary is eager to pry on honest clients' data, and may collude with the server and a fraction of sampled clients to boost its advantages.
We assume mild client collusion because, in real deployments, the number of clients is usually large, making it hard for an adversary to corrupt a large fraction of them.
For example, at the scale of the Apple ecosystem (over 2 billion active devices~\cite{apple2023report}), even compromising 1\% would mean about 20 million nodes.
Consequently, the chance of having many colluded clients sampled by the server is tiny, if it follows the agreed-upon sampling algorithm.
Even if the server behaves maliciously in client sampling, it can still be refrained from impersonating or simulating arbitrarily many clients given the use of a public-key infrastructure and a signature scheme, with which honest clients can verify the source of their received messages and detect such behaviors (\cref{sec:enforcement_consolidation}).
In addition, we can also prevent the server from cherry-picking colluded clients when verifiable random functions (VRFs)~\cite{micali1999verifiable, dodis2005verifiable} are deployed to ensure the correctness of random client sampling (\cref{sec:discussion}).

\begin{figure*}[t]
	\centering
	\begin{subfigure}[b]{0.15\linewidth}
		\centering
		\includegraphics[width=\columnwidth]{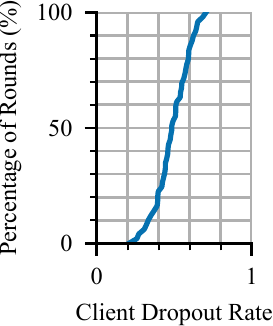}
		\caption{Client dynamics.}
		\label{fig:motivation_privacy_trace_dropout}
	\end{subfigure}\hfill
	\begin{subfigure}[b]{0.29\linewidth}
		\centering
		\includegraphics[width=\columnwidth]{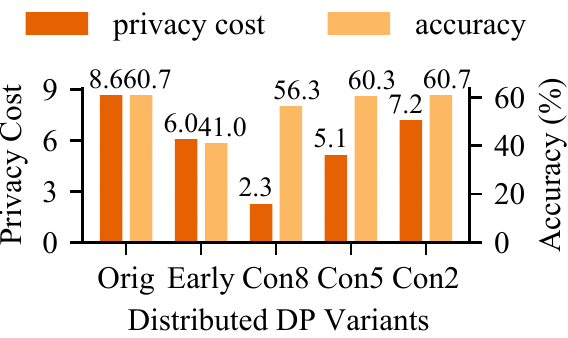}
		\caption{Privacy v.s. utility in CIFAR-10.}
		\label{fig:motivation_privacy_trace_privacy_utility_cifar10}
	\end{subfigure}\hfill
	\begin{subfigure}[b]{0.29\linewidth}
		\centering
		\includegraphics[width=\columnwidth]{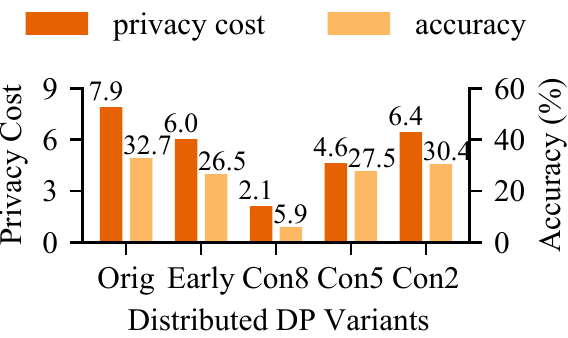}
		\caption{Privacy v.s. utility in CIFAR-100.}
		\label{fig:motivation_privacy_trace_privacy_utility_cifar100}
	\end{subfigure}\hfill
	\begin{subfigure}[b]{0.22\linewidth}
		\centering
		\includegraphics[width=\columnwidth]{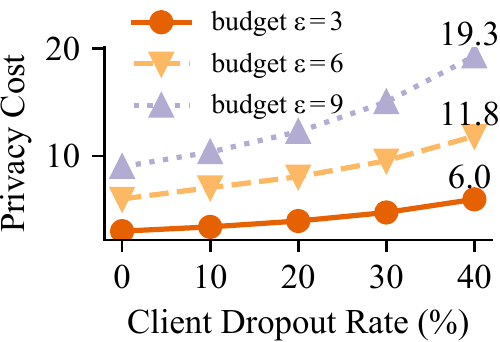}
		\caption{Privacy impact under various client dropout rates.}
		\label{fig:motivation_privacy_simulation_dropout_rate}
	\end{subfigure}
	\caption{Privacy impact of client dropout.}
	\label{fig:motivation_privacy}
\end{figure*}

\subsection{Differential Privacy}
\label{sec:background_dp}

The reconstruction and membership attacks work by finding clients' data that make the observed messages (e.g., individual updates or trained models) more likely.
Differential privacy (DP)~\cite
{cynthia2006differential,dwork2006calibrating,dwork2014algorithmic} effectively prevents these attacks by ensuring that no specific client participation can noticeably
increase the likelihood of such observed messages. This guarantee is captured by two parameters, $\epsilon$ and $\delta$~\cite{dwork2014algorithmic}. Given any
neighboring training sets $D$ and $D^{'}$ that differ only in the
inclusion of a single client's data, the aggregation procedure $f$ is
($\epsilon$, $\delta$)-differentially private if, for any set of
output $R$, we have $\text{Pr}[f(D) \in R] \leq e^{\epsilon} \cdot \text{Pr}[f(D') \in R] + \delta.$

In other words, a change in a client's participation yields at most a
multiplicative change of $e^{\epsilon}$ in the probability of any output, 
except with a probability $\delta$. Intuitively, smaller
$\epsilon$ and $\delta$ indicate a stronger privacy guarantee.
A key property of DP is composition which states that the process of running ($\epsilon_1$, $\delta_1$)-DP and ($\epsilon_2$, $\delta_2$)-DP computations on the same dataset is ($\epsilon_1+\epsilon_2$, $\delta_1+\delta_2$)-DP.
This allows one to account for the privacy loss resulting from a sequence of DP-computed outputs, such as the release of multiple aggregate updates in FL.

\PHM{Central DP and Local DP.}
One way to apply DP in FL is to let the server add DP-compliant noise to
the aggregate update, i.e., the \emph{central DP} scheme~\cite
{kairouz2019advances}. However, the server must be trusted as
it has access to the (unprotected) aggregate update. While the
server may establish a trusted execution environment (TEE)~\cite
{sgx,trustzone} with hardware support, it 
is still vulnerable to various attacks, e.g., side-channel attacks~\cite
{van2018foreshadow, chen2021voltpillager}. An alternative DP scheme is \emph
{local DP}, in which each sampled client adds DP noise to perturb its
local update. As long as the noise added by a client
is sufficient for a DP guarantee on its own, its privacy is
preserved regardless of the behavior of other clients or the server. This, however, results in excessive accumulated noise in the aggregate
update, significantly harming the model utility~\cite
{kairouz2021distributed}.

\PHM{Distributed DP.} 
Compared to central and local DP, \emph{distributed DP} offers an appealing
solution in FL scenarios as it: 1) requires no trusted server, and 2) imposes
minimum noise given a privacy budget.
In distributed DP, a
privacy goal is specified as a global privacy budget $
(\epsilon_G, \delta_G)$, which can be viewed as a 
non-replenishable resource that is consumed by each release of an
aggregate update. Ideally, by the time when the training completes, 
the remaining privacy budget should be zero, so as to meet the privacy goal at
the expense of minimum DP noise and model utility loss.

This requires the system to perform \emph{offline noise planning} ahead of time to determine the minimum required noise 
that should be added to the aggregate update in each training round to control 
the privacy loss.
The system then proceeds to \emph{online noise enforcement}. In each training
round, it evenly splits the noise adding task to all sampled clients. Each of
them slightly perturbs its update by adding an even share of the
minimum required noise. The clients then mask their updates and send them to the
server using the secure aggregation (SecAgg) protocol~\cite
{bonawitz2017practical}, which ensures that the server learns nothing but the
aggregate update that is perturbed with exactly the minimum required noise.
Note that, besides the commonly-used SecAgg, distributed DP can also be implemented using
alternative approaches such as secure shuffling~\cite
{bittau2017prochlo, erlingsson2019amplification, cheu2019distributed}. In
this paper, we focus on the approaches using SecAgg, given their popularity in
FL.

\subsection{Practical Issues of Distributed DP}
\label{sec:background_prior}

While distributed DP can achieve an appealing privacy-utility
tradeoff, its deployment in real world has significant issues.

\subsubsection{Privacy Issue Caused by Client Dropout}
\label{sec:issue_client_dropout}

In FL training, \emph{client dropout} can occur
anytime, e.g., due to low battery, poor connection, or switching to a metered
network.
The prevalence of client dropout, which has been widely observed in
real-world systems~\cite{yang2021characterizing, lai2022fedscale,
bonawitz2019towards}, raises a severe privacy issue in distributed DP.
Specifically, if clients
drop out after being sampled, without their noise contributions, the
total noise added to the aggregate update falls below the minimum required
level. This leads to increased data exposure that forces the system to consume more privacy
budget than planned for each round.
Without the ability to deterministically enforce the consumption of privacy budget, the system may fail to incentivize clients to join or comply with privacy regulations~\cite{piper2000personal, voigt2017eu, pardau2018california}.

To illustrate this issue, we analyze a realistic FL task. We run two FL testbeds in which 100 clients jointly train a ResNet-18~\cite{abadi2016deep} model over the CIFAR-10 dataset and CIFAR-100 dataset~\cite{krizhevsky2009learning} for 150 and 300 rounds, respectively.
To emulate the dynamics of clients, we use a large-scale user behavior dataset spanning 136k mobile devices~\cite{yang2021characterizing} and extract 100 volatile users.
We sample 16 clients to train in each round and observe great dynamics in their availability, as shown in Figure~\ref{fig:motivation_privacy_trace_dropout}.
In this case, the original distributed DP training over CIFAR10 (resp. CIFAR100) with a global privacy budget $\epsilon_G=6$ ends up consuming an $\epsilon$ of $8.6$ (resp. $7.9$) at the 150th (resp. 300th) round due to the missing noise from the dropped clients (see \texttt{Orig} in Figure~\ref{fig:motivation_privacy_trace_privacy_utility_cifar10} and~\ref{fig:motivation_privacy_trace_privacy_utility_cifar100}).

\PHM{Naive Solutions and Limitations.} One solution to this issue is to stop the training early when
the privacy budget runs out (\texttt{Early}). However, this inevitably harms the model
utility.
As shown in Figure~\ref{fig:motivation_privacy_trace_privacy_utility_cifar10} and~\ref{fig:motivation_privacy_trace_privacy_utility_cifar100},
\texttt{Early} reduces the model utility by 19-29\% compared to non-private training.
Another solution is to make a conservative estimation on the per-round dropout rate during offline noise planning.
However, a good estimation that balances the privacy-utility tradeoff requires accurate information on client dynamics and learning tasks.
For instance, without accurate information on client dynamics (e.g., Figure~\ref{fig:motivation_privacy_trace_dropout}), one common practice is to overestimate the dropout severity (e.g., 80\% as in \texttt{Con8}) which leads to suboptimal model utility, while underestimating it (e.g., 20\% as in \texttt{Con2}) results in excessive privacy budget consumption.
Even given a priori knowledge of client dynamics, the trade-off is still hard to navigate without trial-and-error experiments due to its task-specific nature.
For example, while guessing 50\% as the per-round dropout rate (\texttt{Con5}) yields a near-optimal privacy-utility tradeoff for CIFAR-10, it causes noticeable utility degradation (16\%) for CIFAR-100.

\PHM{Impact of Client Dropout Rate.}
To further relate privacy violation to dropout severity, we let the clients randomly drop with a configurable
rate after being
sampled.
Figure~\ref{fig:motivation_privacy_simulation_dropout_rate} shows that, for the CIFAR-10 testbed, as the dropout rate
increases, more clients' data gets exposed during training, leading to a larger privacy deficit regardless of the budget.

\begin{figure}[t]
	\begin{subfigure}[b]{1.0\columnwidth}
		\centering
		\includegraphics[width=\columnwidth]{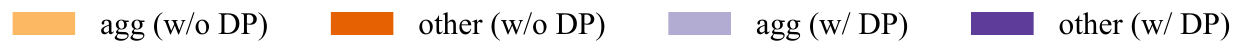}
	\end{subfigure}
	\begin{subfigure}[b]{0.49\columnwidth}
		\centering
		\includegraphics[width=\columnwidth]{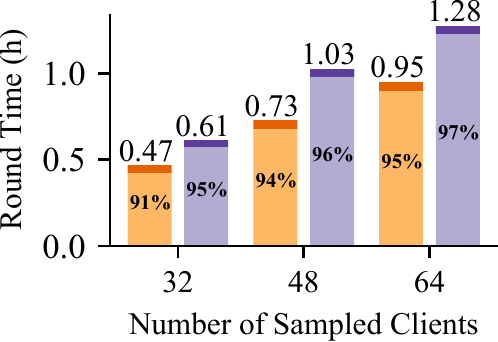}
		\caption{SecAgg with per-round dropout rate being 10\%.}
		\label{fig:motivation_testbed_efficiency_2}
	\end{subfigure}
	\begin{subfigure}[b]{0.49\columnwidth}
		\centering
		\includegraphics[width=\columnwidth]{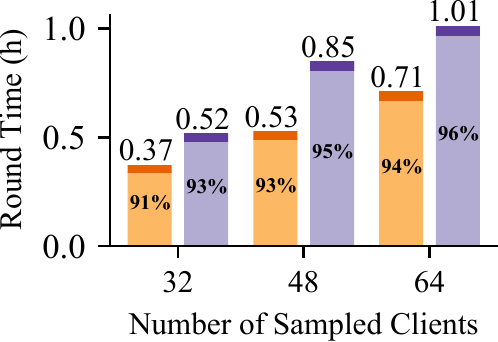}
		\caption{SecAgg+ with per-round dropout rate being 10\%.}
		\label{fig:motivation_testbed_efficiency_4}
	\end{subfigure}
	\caption{Impact of secure aggregation on training efficiency.}
	\label{fig:motivation_testbed}
\end{figure}

\subsubsection{Performance Issue Caused by Secure Aggregation}
\label{sec:secagg_issue}

Besides the privacy issue caused by client dropout, the SecAgg
algorithm~\cite{bonawitz2017practical} used in distributed DP creates a severe performance issue.
Specifically, to
ensure the server learns no individual update from any client but the
aggregate update only, SecAgg lets clients synchronize secret keys and use
them to generate zero-sum masks (a detailed description of SecAgg is embedded in Figure~\ref{fig:enforcement_integration}).
This involves extensive
use of pairwise masking and secret sharing, incurring high complexity in computation and communication.

To quantify the performance impact of SecAgg, we refer to Figure~\ref{fig:motivation_testbed_efficiency_2} which shows the breakdown of the average
runtime of one training round in the previous experiments with varying numbers of sampled clients.
For comparison, we also run the experiment with SecAgg but
add no DP noise to the aggregate update.
In all experiments, the  cost of SecAgg dominates, accounting for 86-91\%
of the training time, while SecAgg with DP features a slightly more serious bottleneck than that without DP.
Furthermore, the dominance of SecAgg is accentuated at scale.

\PHM{Existing Solutions and Limitations.} 
There have been active studies on improving the performance of secure aggregation, but they all have significant limitations.
For example, one guarantee provided by SecAgg is input privacy against malicious adversaries, while TurboAgg~\cite{so2021turbo}, FastSecAgg~\cite{kadhe2020fastsecagg}, and LightSecAgg~\cite{so2022lightsecagg} only handle a semi-honest adversary.
Moreover, their improved complexity comes at the cost of degraded dropout tolerance~\cite{jahani2023swiftagg+, liu2023dhsa}, with their communication cost still being high in FL practice~\cite{ma2023flamingo}.
Among the follow-up works of SecAgg, SecAgg+~\cite{bell2020secure} is the state-of-the-art which improves the asymptotic complexity with a slight compromise on security and robustness.
Yet, as Figure~\ref{fig:motivation_testbed_efficiency_4} implies, a further improvement is still desired given the consistent dominance of SecAgg+ on the training time. 

\section{Dropout-Resilient Noise Enforcement}
\label{sec:enforcement}

To tackle the privacy issue mentioned in~\cref{sec:issue_client_dropout}, we first formalize the noise enforcement problem under client dropout, and present the technical intuition to address this problem (\cref{sec:enforcement_intuition}). We then
describe a novel `add-then-remove' noise enforcement approach that realizes this intuition (\cref{sec:enforcement_add}) with security consolidation in real deployments (\cref{sec:enforcement_consolidation}),
followed by a security analysis (\cref{sec:enforcement_security}).
Without loss of generality, we assume that the random noise distribution $\chi
(\sigma^2)$ used in DP is \emph{closed under summation} w.r.t. the variance
$\sigma^2$. That is, given two independent noises $X_1 \sim \chi(\sigma_1^2)$ and
$X_2 \sim \chi(\sigma_2^2)$,
we have
$X_1 \pm X_2 \sim \chi(\sigma_1^2 + \sigma_2^2)$. For example, both Gaussian and 
Skellam~\cite{agarwal2021skellam} distributions exhibit this property. 

\subsection{Technical Intuition}~\label{sec:enforcement_intuition}

We start with a formal description of the
original noise addition process used in distributed DP, denoted as \texttt{Orig}.
\begin{definition}[\textup{\texttt{Orig}}]
    Given the set of sampled clients $U$ and the target noise level $\sigma^2_*$
    in a certain round, \textup{\texttt{Orig}} lets each client $c_j \in U$ perturb
    its update $\bm{\Delta}_j$ by adding noise $\bm{n}_j \sim \chi(\sigma^2_* / \lvert U \rvert)$ and upload the result $\tilde{\bm{\Delta}_j} =\bm{\Delta}_j + \bm{n}_j$ to the server for aggregation.
    \label{def:orig}
\end{definition}

As described in \cref
{sec:issue_client_dropout}, the problem with \texttt{Orig} is that when some clients drop after being sampled, their noise contributions are missing; thus, the eventual noise aggregated at the server
will be insufficient.
One potential fix is to let the server add back the
missing noise contributed by the dropped clients~\cite
{kairouz2021distributed}.
However, this is not viable under our threat model (\cref{sec:background_threat}) in which the server can be part of the adversary who can infer clients' data from the insufficiently perturbed aggregate result with unbounded advantages (semi-honest) and/or even omit the noise addition task (malicious).

\PHM{Add-Then-Remove Noise Enforcement.}
To ensure that the aggregate noise never goes inefficient and always lands at the minimum required level, we design an `add-then-remove' noise enforcement scheme: 1) each sampled client first adds a higher-than-required amount of noise  to its model update, rendering an overly perturbed aggregate update despite client dropout, and
2) the server removes the excessive part of the aggregate noise based on  the actual dropout outcome.
This scheme can be realized by two possible approaches:

\begin{figure}[t]
	\centering
	\begin{subfigure}[b]{0.485\columnwidth}
		\centering
		\includegraphics[width=\columnwidth]{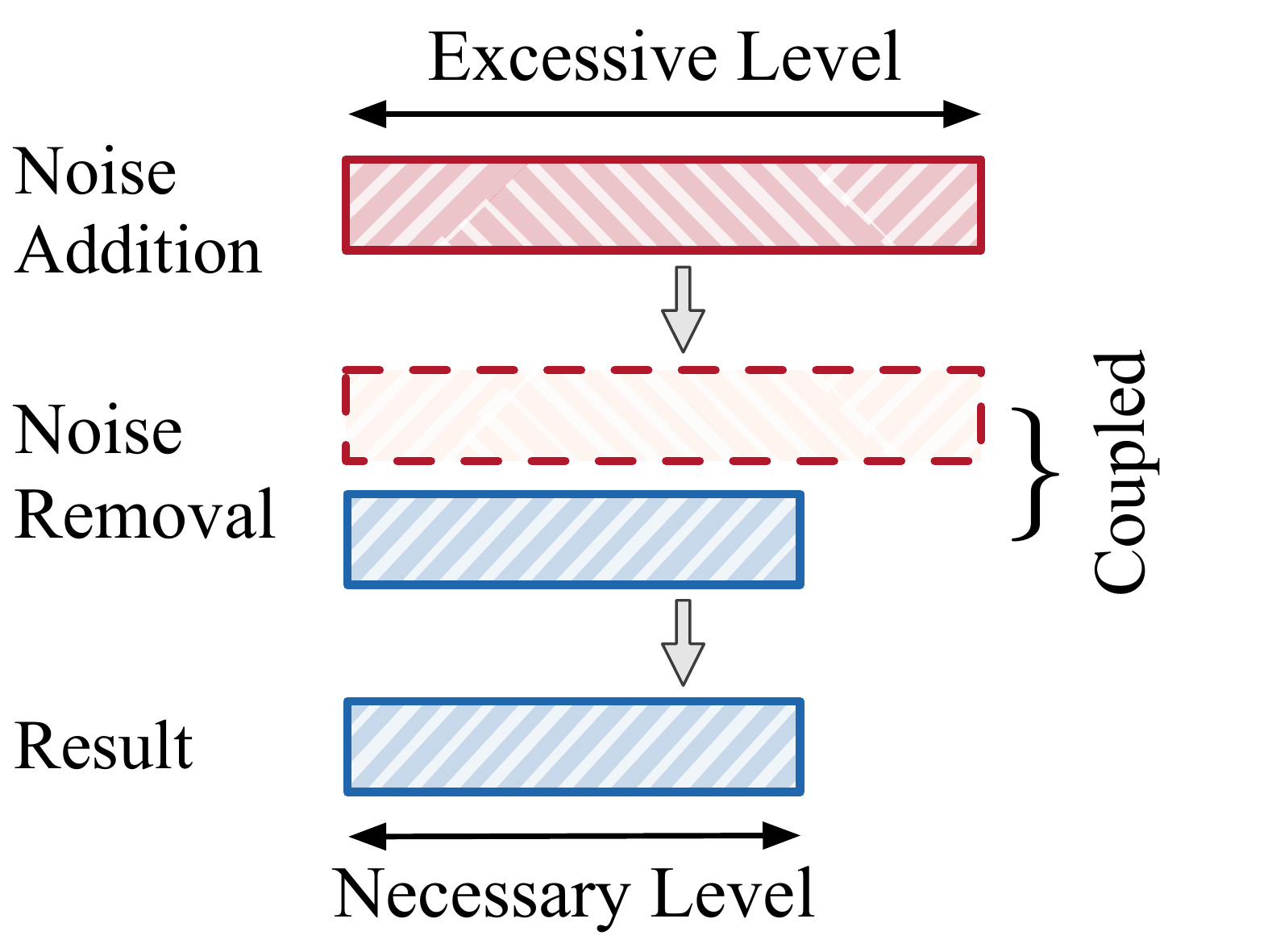}
		\caption{Rebasing.}
		\label{fig:enforcement_add_strawman}
	\end{subfigure} \hfill
	\begin{subfigure}[b]{0.495\columnwidth}
		\centering
		\includegraphics[width=\columnwidth]{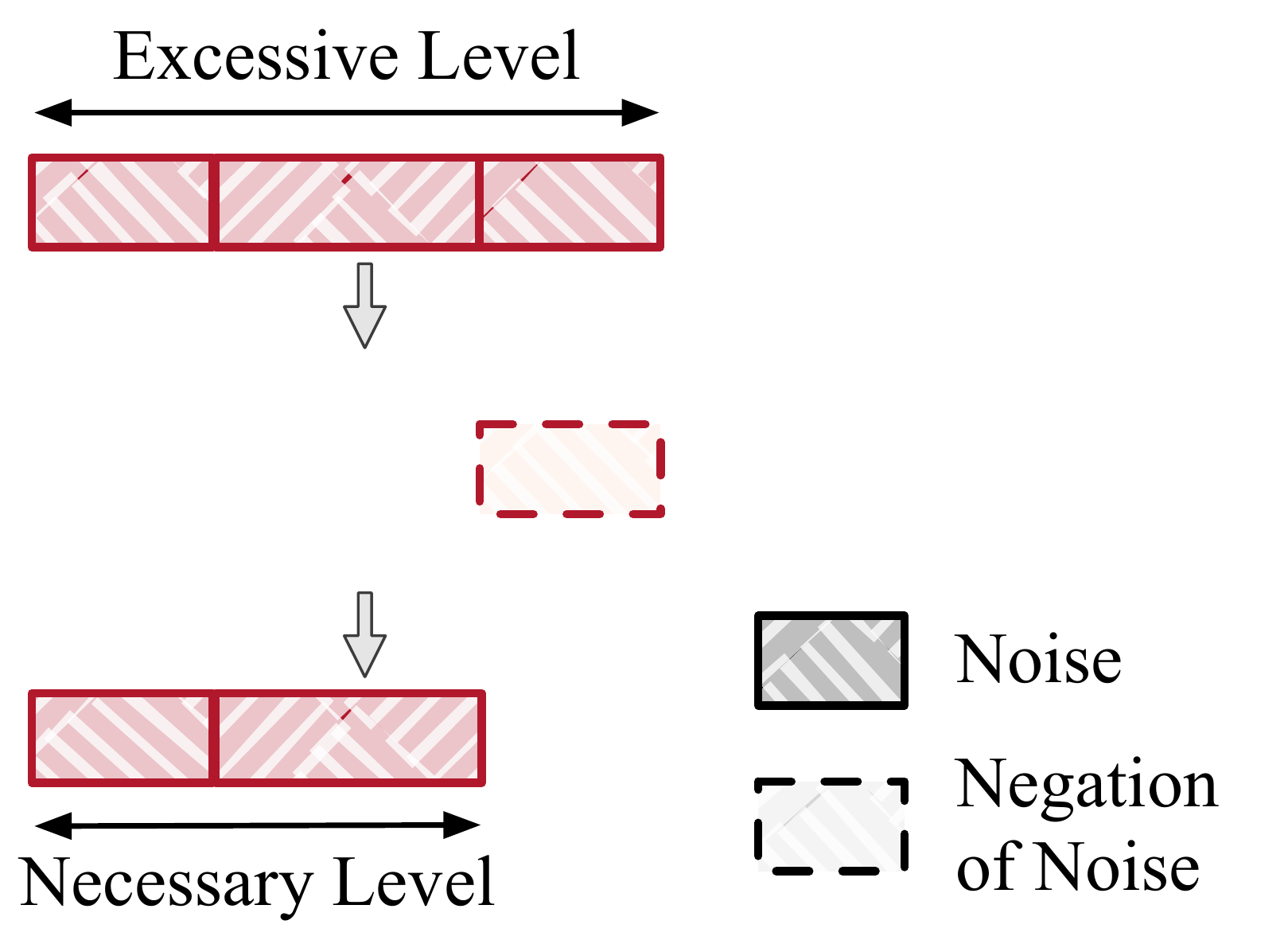}
		\caption{Decomposition.}
		\label{fig:enforcement_add_separate}
	\end{subfigure}
	\caption{Two `add-then-remove' approaches.}
	\label{fig:enforcement_add}
\end{figure}

\begin{figure*}[t]
	\captionsetup[subfigure]{position=b}
	\centering
	\setbox9=\hbox{\includegraphics[width=.36\linewidth]{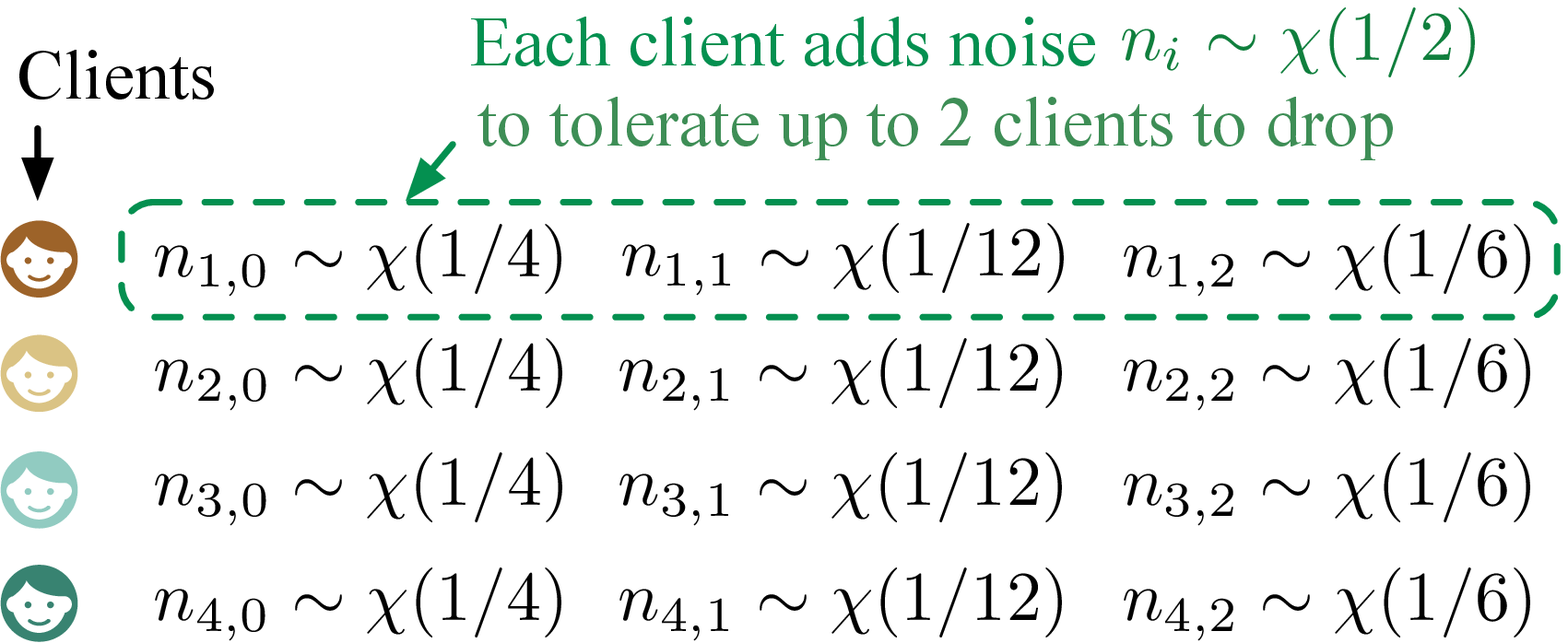}}
	\subcaptionbox{Noise decomposition.\label{fig:enforcement_xnoise_prec_add}}{\includegraphics[width=.35\linewidth]{./images/xnoise-prec-detailed.eps}}\hfill
	\subcaptionbox{No client drops.\label{fig:enforcement_xnoise_prec_0}}{\raisebox{\dimexpr\ht9-\height}{\includegraphics[width=.20\linewidth]{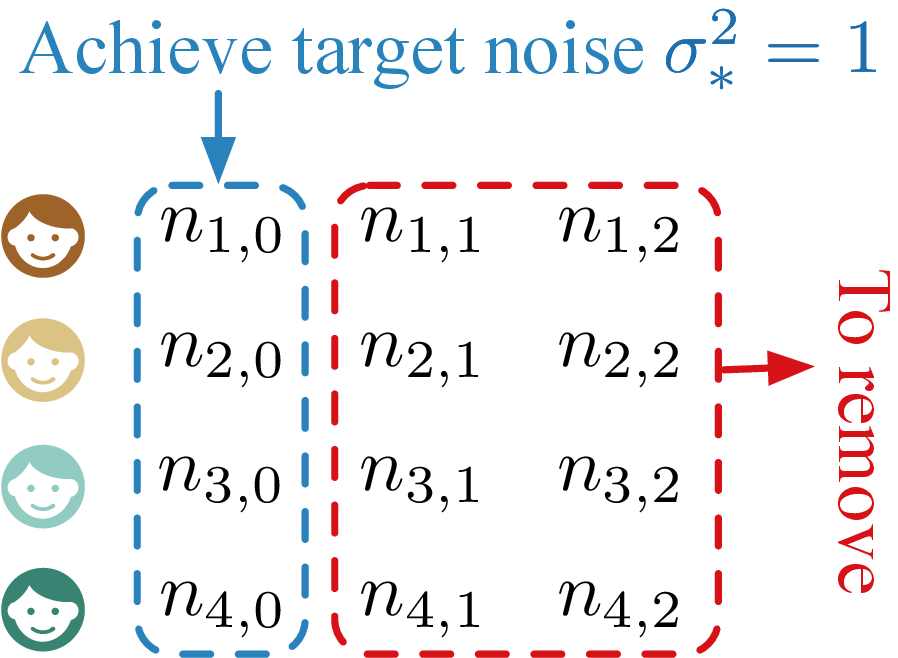}}}\hfill
	\subcaptionbox{One client drops.\label{fig:enforcement_xnoise_prec_1}}{\raisebox{\dimexpr\ht9-\height}{\includegraphics[width=.20\linewidth]{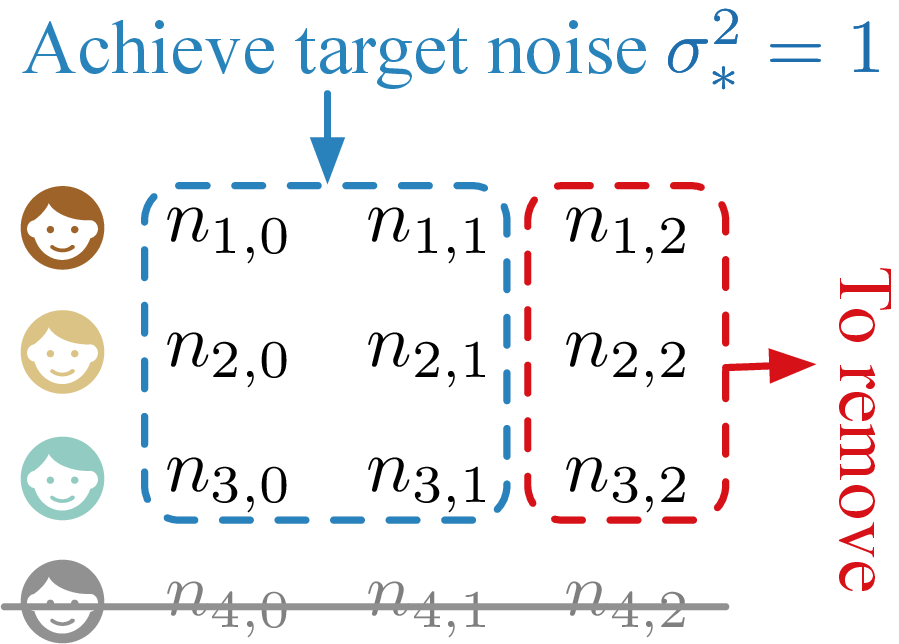}}}\hfill
	\subcaptionbox{Two clients drop.\label{fig:enforcement_xnoise_prec_2}}{\raisebox{\dimexpr\ht9-\height}{\includegraphics[width=.20\linewidth]{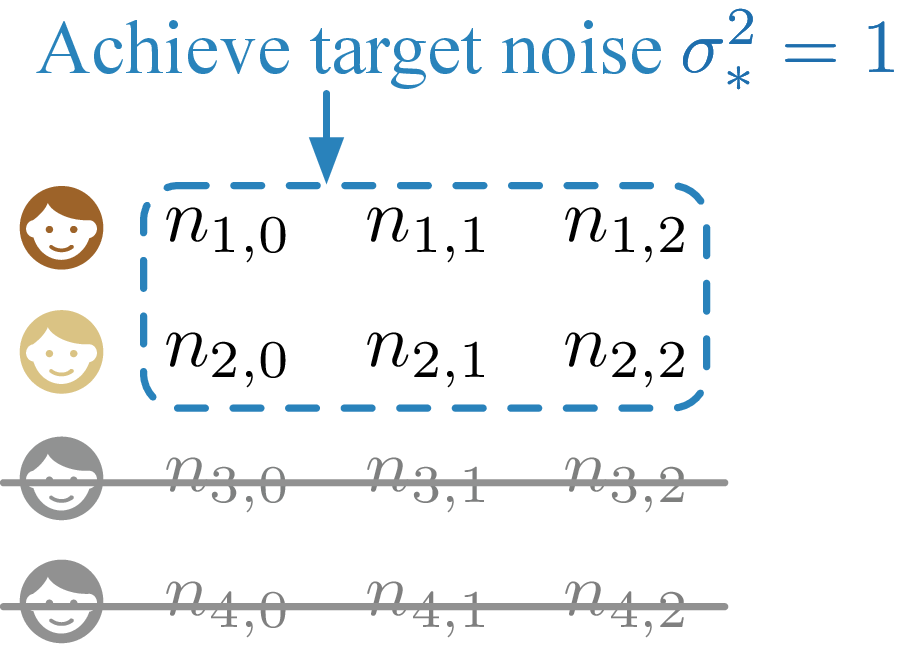}}}
	\caption{An illustration of how the `add-then-remove' approach with noise decomposition deals with client dropout precisely.}
	\label{fig:enforcement_xnoise_prec}
\end{figure*}

\begin{itemize}
    \item \emph{Rebasing}: During noise addition, each sampled client adds its noise share $\bm{n}_o$ to the local update, and sends the noisy update as a whole to the server.
    To facilitate noise removal, each surviving client computes the newly-required noise $\bm{n}_u$ based on the actual client dropout outcome, and subtracts the original noise share $\bm{n}_o$.
    To ensure that the noise removal is privacy-preserving, only `$\bm{n}_u - \bm{n}_o$' is transmitted to the server and added to the aggregate update (Figure~\ref{fig:enforcement_add_strawman}).
    This approach was adopted by~\cite{baek2021enhancing}.
    \item{\emph{Decomposition:} Instead of adding noise as a whole, each sampled client decomposes its noise share into multiple additive components that can be  added separately to the client's local update.
    For privacy-preserving noise removal, each surviving client only sends  the noise components that are over the newly-required amount to the server for subtracting them from the aggregate (Figure~\ref{fig:enforcement_add_separate}).}
\end{itemize}

\PHM{Comparison.}
One difference between the two approaches lies in their \emph{communication efficiency}.
In FL, a DP noise is a sequence of pseudo-random numbers (PRNs) of the same length (e.g., millions to billions) as the model,
and can be uniquely generated via feeding a seed (e.g., 20 bytes) into a PRN generator.
During noise removal, `decomposition' allows each surviving client to send the relevant seeds to the server for it to generate each noise component that needs to be removed.
However, `rebasing' requires each surviving client to generate and send the updated noise $\bm{n}_u - \bm{n}_o$ as a whole to conceal the two individual noises.
Otherwise, the server can use them to reconstruct the noise-free aggregated update. 
This results in poor efficiency as the communication cost of noise removal increases prohibitively with the ever-growing model size.
In Section~\ref{sec:evaluation_efficiency_noise}, we compare the scalability of the two approaches in communication.

Another difference lies in their \emph{robustness}.
In reality, surviving clients can also drop out in the middle of noise removal.
Missing their noise seeds, the server cannot fully remove the excessive noise added to the aggregate update.
To tackle this issue, `decomposition' can efficiently back up each noise component before model aggregation by secret-sharing its seed across clients  (e.g., via the Shamir's scheme~\cite{shamir1979share}).
Such a scheme, however, does not apply to `rebasing' as the updated noise to transmit can neither be generated with a seed nor determined before aggregation.

\PHM{Challenges.} Due to the poor efficiency and robustness of the `rebasing' approach, we opt to instantiate the `add-then-remove' noise enforcement scheme with `decomposition' and tackle two technical challenges in its design:
\begin{itemize}
	\item Given various dropout outcomes of a training round, how to decompose a client's noise share to accommodate every possible requirement during noise removal (\cref{sec:enforcement_add})?
	\item In real deployments, how to make the noise enforcement approach secure, preferably in an efficient way (\cref{sec:enforcement_consolidation})?
\end{itemize}

\subsection{Add-Then-Remove with Noise Decomposition}
\label{sec:enforcement_add}

We start with how much noise a sampled client should add.
Without loss of generality, we assume that
the system's \emph{tolerance} to client dropout is a configurable parameter.
Let $U$ be the set of sampled clients in a certain training round, among which the
system can tolerate up to $T$ dropouts.

\PHM{Noise Addition and Removal.}
Let $\sigma^2_*$ be the target noise
level in each round. To meet this noise level even in the worst case, 
each client in \system{} adds an excessive noise at the
level of $\frac {\sigma^2_*}{\lvert U \rvert - T}$. In doing so, even if there
are $T$ clients dropped after being sampled, the total noise contributed by
the surviving $\lvert U \rvert - T$ clients is still sufficient at the target level.

On the other hand, when fewer than $T$ clients dropped after being sampled, the aggregate noise exceeds the target level and part of it needs to be removed for model utility. Let $D \subset U$ denote the set of
clients dropped after being sampled, where $\lvert D \rvert
\le T$. The amount of excessive noise that should be removed by the server is:
\begin{equation}
    l_\mathrm{ex} = \underbrace{(\lvert U \rvert - \lvert D \rvert)\frac{\sigma^2_*}{\lvert U \rvert - T}}_{\textup{Actual noise level}} - \sigma^2_* = \frac{T - \lvert D \rvert}{\lvert U \rvert - T}\sigma^2_*.
    \label{eq:noise_to_drop}
\end{equation}

\system{} evenly distributes the noise removal task across surviving clients, therefore, each of them needs to help the server remove the noise of level:
\begin{equation}
    l'_\mathrm{ex}=\frac{l_\mathrm{ex}}{\lvert U \rvert - \lvert D \rvert} = \sigma^2_* \Bigl( \frac{1}{\lvert U \rvert - T} - \frac{1}{\lvert U \rvert - \lvert D \rvert} \Bigr).
    \label{eq:local_noise_to_drop}
\end{equation}

\PHM{Noise Decomposition for Precise Control.}
Equation~\eqref{eq:local_noise_to_drop} indicates that the noise to be removed by a surviving client decreases when the number of dropped clients increases.
Given this monotonicity, each client in \system{}
can carefully decompose its added noise share into multiple additive components, and remove some of such noise components when needed for the precise control of noise level.

For example, consider a scenario where the number of sampled clients $\lvert U \rvert = 4$, the dropout tolerance $T=2$, and the target noise level $\sigma_*^2=1$.
To enforce this target even if 2 clients drop out, the noise added to each client's update should be of the level 1/2.  Moreover, 
as shown in Figure~\ref{fig:enforcement_xnoise_prec_add}, such a noise can be added  as 3 separate components of level 1/4, 1/12, and 1/6, respectively, then one can accommodate all possible dropout outcomes within the tolerance by subtracting a subset of the added components to precisely remove the excessive noise.  To be exact, if no client drops, i.e., $\lvert D \rvert = 0$, each surviving client removes $l_{ex}'=1/12+1/6$ (Figure~\ref{fig:enforcement_xnoise_prec_0}); if one client drops, i.e.,  $\lvert D \rvert = 1$, each surviving client removes $l_{ex}'=1/6$  (Figure~\ref{fig:enforcement_xnoise_prec_1}); if two clients drop, i.e., $\lvert D \rvert = 2$, each surviving client removes $l_{ex}'=0$ (Figure~\ref{fig:enforcement_xnoise_prec_2}).

To be general, \system{} lets each client $c_i \in U$ decompose the added noise $\bm{n}_i \sim \chi(\frac{\sigma^2_*}{\lvert U \rvert - T})$ into $T + 1$ components, i.e., $\bm{n}_i = \sum_{k=0}^T \bm{n}_{i,k}$ where $\bm{n}_{i,0} \sim \chi(\frac{\sigma^2_*}{\lvert U \rvert})$ and $\bm{n}_{i,k} \sim \chi(\frac{\sigma^2_*}{(\lvert U \rvert - k + 1)(\lvert U \rvert - k)})$ for $k = 1, 2, \cdots, T$.
These noise components are constructed in a way that when there are $\lvert D \rvert$ clients dropped after being sampled, the noise components $\bm{n}_{i,k}$ contributed by the surviving clients $c_i \in U \setminus D$ with the index $k > \lvert D \rvert$ become excessive and should be removed.
One can verify that the aggregate of these removed components is exactly $l_\mathrm{ex}$, i.e., $\sum_{c_i \in U \setminus D}\sum_{k = \lvert D \rvert + 1}^T \bm{n}_{i, k} \sim \chi(l_\mathrm{ex})$.
We formalize the noise enforcement process for precise noise control as \textup{\texttt{XNoise}}.

\begin{definition}[\textup{\texttt{XNoise}}]
    In each training round, a sampled client $c_i \in U$ adds the intended
    excessive noise to its update $\bm{\Delta}_i$ and sends the perturbed result
    $\tilde{\bm{\Delta}}_i = \bm{\Delta}_i + \sum_{k=0}^T \bm{n}_{i, k}$ to the server.
    Among these sampled clients, a subset $D$ has dropped where $\lvert D \rvert \le T$. 
    The server calculates the aggregate update
    $\tilde{\bm{\Delta}} = \sum_{c_i \in U \setminus D} \tilde{\bm{\Delta}}_i$, and then
    removes some excessive noise components contributed by the surviving clients 
    (known as \emph{survivals}) to precisely enforce the target noise level, i.e.,
    $\tilde{\bm{\Delta}} - \sum_{c_i \in U \setminus D}\sum_{k = \lvert D \rvert + 1}^T \bm{n}_{i, k}$.
    \label{def:xnoise-prec}
\end{definition}

\PHM{Dropout-Resilient Noise Removal with Secret Sharing.}
As described in~\cref{sec:enforcement_intuition}, our noise decomposition design allows clients to transmit the seeds that are used to generate the requested noise components instead of those components in noise removal, greatly reducing the communication overhead.
To further make this process dropout-resilient,  we use Shamir's secret sharing scheme~\cite
{shamir1979share} for seed bookkeeping: each sampled client secretly shares with others the seeds it uses to generate local noise components before the secure aggregation takes place.

As such, to recover a local noise component during noise removal, the server first directly consults the related client on the corresponding seed.
If the client drops out before reporting the seed, the server then initiates one additional communication round to collect the secret shares of the seed from all available clients. The server can recover the seed provided that the number of responding clients in this communication round exceeds a certain threshold $\tau$ specified by the secret sharing scheme~\cite{shamir1979share}.

Given all the above, a faithful execution of \texttt{XNoise} strictly 
enforces the target noise level, as established by Theorem~\ref{thm:xnoise-prec}. The proof is
given in Appendix~\ref{sec:appendix_xnoise}.

\begin{theorem}[Correctness]
	\texttt{XNoise} ensures the noise level in the aggregate update is exactly $\sigma^2_*$, regardless of the dropout outcome as long as $\lvert D \rvert \leq T$, i.e., the number of dropped clients does not exceed the dropout tolerance.
	\label{thm:xnoise-prec}
\end{theorem}

\subsection{Security Consolidation with Optimized Practice}
\label{sec:enforcement_consolidation}

We next describe how we deploy \texttt{XNoise} in unsecure environments with optimized implementation.

\PHM{Establishment of Secure Channels across Clients.}
For a client to be able to share a secret with another client as desired in~\cref{sec:enforcement_add}, they first establish a secure channel.
To achieve this over a server-mediated network, \system{} has them conduct key agreement via the Diffie-Hellman protocol~\cite{merkle1978secure} to establish the shared secret key for encrypting the subsequent communication.
Furthermore, when a channel needs to be authenticated (under the malicious threat model), either end of the channel has to sign its messages for the other end to verify its identity with a public key infrastructure (PKI).


\PHM{Prevention from Understating Dropout.}
As characterized by Equation~\eqref{eq:noise_to_drop},
the less severe the client dropout is, the more noise the server removes.
While a semi-honest server faithfully runs the protocol and never reports a smaller number of dropped clients than the actual one,
a malicious server may \emph{understate} the dropout severity for removing more noise than needed and obtaining an insufficiently perturbed aggregate.
In the worst case, the malicious server can reduce the aggregate noise level to $(1 - T/  \lvert U \rvert)\sigma^*$, e.g., only $40\%$ of the target noise remains given the dropout tolerance set as $60\%$ of the sampled clients.
To detect whether a malicious server understates client dropout,  \system{} lets clients verify the broadcasted dropout outcome with the use of a PKI and a standard UF-CMA\footnote{Unforgability against Chosen-Message Attacks.} signature scheme \textbf{SIG}:
\begin{itemize}
\item Before uploading its perturbed local update $\tilde{\Delta}_i$, each client $i$ signs the current round number $R$ with its signing key $d^{SK}_i$ and produces a signature $\omega'_i \leftarrow \textrm{\textbf{SIG.sign}}(d^{SK}_i, R)$. $\omega'_i$  is sent along with $\tilde{\Delta}_i$ to the server.
\item When broadcasting the dropout outcome $D$, the server also broadcasts the set $\{j, \omega'_j\}_{j \in P}$, which contains all the signatures it has received ($P$ denoted as the related clients).
\item After receiving $D$ and $\{j, \omega'_j\}_{j \in P}$, each client $i$ verifies that: 1) all signatures are correct, i.e., $\textrm{\textbf{SIG.ver}}(d^{PK}_j, R, \omega'_j) = 1$ for all $j \in P$, and 2) they agree with the broadcasted dropout outcome, i.e., $P = U \setminus D$ (otherwise aborts).
\end{itemize}
Intuitively, for the server to pretend that a client $j$ survives, it has to forge that client's signature on the current round number, which is computationally infeasible if the client in fact dropped out, given the security of the signature scheme.

\PHM{Optimization via Integration with Secure Aggregation.} 
The aforementioned security-related secure channel establishment and the dropout outcome verification both induce $O(\lvert U \rvert)$ cost to each sampled client in computation and communication, and $O(\lvert U \rvert^2)$ cost to the server in communication.
On the other hand, secure aggregation, the other indispensable component in the distributed DP workflow, often has instantiated similar primitives for correctness and security of its execution~\cite{bonawitz2017practical, bell2020secure}.
We thus repurpose the existing security infrastructure to reduce the implementation complexity and improve the runtime efficiency.

\begin{table*}
	\footnotesize
	\begin{tabularx}{\linewidth}{X}
		\toprule
		\\
		\centerline{The SecAgg Protocol integrated with \texttt{XNoise}} \\
		\vspace*{-0.2in}
		\begin{itemize}
			\item \textbf{Setup}:
			\begin{itemize}
				\item[--] All parties are given the current round index $r$, \reuse{the security parameter $\eta$, the number of users $\lvert U \rvert$ and a threshold for SecAgg $t$, honestly generated $pp \leftarrow \textrm{\textbf{KA.gen}}(\eta)$, parameters $m$ and $R$ such that $\mathbb{Z}^m_R$ is the space from which inputs are sampled, and a field $\mathbb{F}$ to be used for secret sharing }\xnoise{and noise samping}, \xnoise{a noise distribution $\chi$, the target central noise level $\sigma_*^2$, a dropout tolerance $T$ and a collusion tolerance $T_C$ for \texttt{XNoise}}. \reuse{All users also have a private authenticated channel with the server.}
			\end{itemize}
			\(
			\textrm{\textit{User} \ } u \left\{
			\begin{minipage}[c]{0.9\linewidth}
				\begin{itemize}
					\item[--] \xnoise{Have an input vector $\bm{\Delta}_u$,
						noises $\bm{n}_{u, 0} \stackrel{g_{u, 0}}{\sim} \chi \left(\frac{\sigma_*^2}{\lvert U \rvert} \cdot \frac{t}{t - T_C} \right)$,
						$\bm{n}_{u, k} \stackrel{g_{u, k}}{\sim} \chi \left(\frac{\sigma_*^2}{(\lvert U \rvert - k + 1) (\lvert U \rvert - k)} \cdot \frac{t}{t - T_C} \right)$ ($k\in[1, T]$) where $g_{u, k} \leftarrow \mathbb{F}$ for all $k$.}
					\item[--] \xnoise{Add to $\bm{\Delta}_u$  a series of noises and produce $\tilde{\bm{\Delta}}_u$: $\tilde{\bm{\Delta}}_u = \bm{\Delta}_u + \sum_{0 \leq k \leq T} \bm{n}_{u, k}$.} 
					\item[--] \malicious{\reuse{Receive signing key $d^{SK}_u$ from the trusted third party, together with verification keys $d^{PK}_v$ bound to each identity $v$.}}
				\end{itemize}
			\end{minipage}
			\right.
			\)
			\item \textbf{Stage 0 (AdvertiseKeys)}: \newline
			\(
			\textrm{\textit{User} \ } u \left\{
			\begin{minipage}[c]{0.9\linewidth}
				\begin{itemize}
					\item[--] Generate key pairs \reuse{$(c^{PK}_u, c^{SK}_u) \leftarrow \textrm{\textbf{KA.gen}}(pp)$}, $(s^{PK}_u, s^{SK}_u) \leftarrow \textrm{\textbf{KA.gen}}(pp)$, \malicious{\reuse{and generate $\omega_u \leftarrow \textrm{\textbf{SIG.sign}}(d_u^{SK}, c_u^{PK} \Vert s_u^{PK}) = 1$}}.
					\item[--] Send $($\reuse{$c_u^{PK}$}$\Vert s_u^{PK}$\malicious{\reuse{$\Vert \omega_u$}}$)$ to the server (through the private authenticated channel) and move to next round.
				\end{itemize}
			\end{minipage}
			\right.
			\)
			\(
			\textrm{\textit{Server} \ } \left\{
			\begin{minipage}[c]{0.9\linewidth}
				\begin{itemize}
					\item[--] \reuse{Collect at least $t$ messages from individual users in the previous round (denote with $U_1$ this set of users). Otherwise, abort.}
					\item[--] Broadcast to all users in $U_1$ the list $\{($\reuse{$v, c_v^{PK}$}$, s_v^{PK}$\malicious{\reuse{$, \omega_v$}}$\}_{v \in U_1}$.
				\end{itemize}
			\end{minipage}
			\right.
			\)
			\item \textbf{Stage 1 (ShareKeys)}: \newline
			\(
			\textrm{\textit{User} \ } u \left\{
			\begin{minipage}[c]{0.9\linewidth}
				\begin{itemize}
					\item[--] Receive the list $\{($\reuse{$v, c_v^{PK}$}$, s_v^{PK}$\malicious{\reuse{$, \omega_v$}}$\}_{v \in U_1}$ broadcasted by the server. \reuse{Assert that $\lvert U_1 \rvert > t$, that all the public keys pairs are different}, \malicious{\reuse{and that $\forall v \in U_1$, $\textrm{\textbf{SIG.ver}}(d_v^{PK}, c_v^{PK} \Vert s_v^{PK}, \omega_v) = 1$}}.
					\item[--] Sample a random element $b_u \leftarrow \mathbb{F}$ (to be used as a seed for a $\textrm{\textbf{PRG}}$).
					\item[--] Generate $t$-out-of-$\lvert U_1 \rvert$ shares of $s_u^{SK}$: $\{(v, s_{u, v}^{SK})\}_{v \in U_1} \leftarrow \textrm{\textbf{SS.share}}(s_u^{SK}, t, U_1)$ and of $b_u$: $\{(v, b_{u, v})\}_{v \in U_1} \leftarrow \textrm{\textbf{SS.share}}(b_u, t, U_1)$.
					\item[--] \xnoise{Generate $t$-out-of-$\lvert U_1 \rvert$ shares for each of $g_{u, k}$ where $k > 0$: $\{(v, g_{u, k, v})\}_{v \in U_1} \leftarrow \textrm{\textbf{SS.share}}(g_{u, k}, t, U_1)$ for $k = [1, T]$.}
					\item[--]For each other user $v \in U_1 \setminus \{u\}$, compute \reuse{$e_{u, v} \leftarrow \textrm{\textbf{AE.enc}}( $ $\textrm{\textbf{KA.agree}}(c_u^{SK}, c_v^{PK})$}$, u \Vert v \Vert s_{u, v}^{SK} \Vert b_{u,v}$\xnoise{$\Vert \{ g_{u,  k, v}\}_{1\leq k \leq T}$}$)$.
					\item[--] \reuse{If any of the above operations (assertion, signature verification, key agreement, encryption) fails, abort.}
					\item[--] \reuse{Send all the ciphertexts $e_{u, v}$ to the server (each implicitly containing addressing information, $u$, $v$ as metadata)}.
				\end{itemize}
			\end{minipage}
			\right.
			\)
			\(
			\textrm{\textit{Server} \ } \left\{
			\begin{minipage}[c]{0.9\linewidth}
				\begin{itemize}
					\item[--] \reuse{Collect lists of ciphertexts from at least $t$ users (denote with $U_2 \subseteq U_1$ this set of users)}.
					\item[--] \reuse{Sends to each user $u \in U_2$ all ciphertexts encrypted for it: $\{e_{u,v}\}_{v \in U_2}$}.
				\end{itemize}
			\end{minipage}
			\right.
			\)
			\item \textbf{Stage 2 (MaskedInputCollection)}: \newline
			\(
			\textrm{\textit{User} \ } u \left\{
			\begin{minipage}[c]{0.9\linewidth}
				\begin{itemize}
					\item[--] \reuse{Receive (and store) from the server the list of ciphertexts $\{e_{u, v}\}_{v \in U_2}$ (and infer the set $U_2$). If the list is of size $<t$, abort.}
					\item[--] For each other user $v \in U_2 \setminus \{u\}$, compute $s_{u, v} \leftarrow \textrm{\textbf{KA.agree}}(s_u^{SK}, s_v^{PK})$ and expand this value using a $\mathrm{\mathbf{PRG}}$ into a random vector $\bm{p}_{u, v} = \gamma_{u, v} \cdot \textrm{\textbf{PRG}}(s_{u,v})$, where $ \gamma_{u, v} = 1$ when $u > v$, and $\gamma_{u, v} = -1$ when $u<v$ (note that $\bm{p}_{u, v} + \bm{p}_{v, u} = 0 \forall u \neq v$). Define $\bm{p}_{u, u} = 0$.
					\item[--] Compute the user's own private mask vector $\bm{p}_u = \textrm{\textbf{PRG}}(b_u)$ and the masked perturbed input $\bm{y}_u \leftarrow \tilde{\bm{\Delta}}_u + \bm{p}_u + \sum_{v \in U_2} \bm{p}_{u, v} \pmod{R}$.
					\item[--] If any of the above operations (key agreement, PRG) fails, abort. Otherwise, send $\bm{y}_u$ to the server and move to the next round.
				\end{itemize}
			\end{minipage}
			\right.
			\)
			\textit{Server}: Collect $\bm{y}_u$ from at least $t$ users (denote with $U_3 \subseteq U_2$ this set of users). Send to each user in $U_3$ the list $U_3$.
			\item \malicious{\reuse{\textbf{Stage 3 (ConsistencyCheck)}}}: \newline
			\textit{User} $u$: \malicious{\reuse{Receive from the server a list $U_3 \subseteq U_2$ containing at least $t$ users ($u$ included).
					Abort if $\lvert U_3 \rvert < t$.
					Send to the server $\omega'_u \leftarrow \textrm{\textbf{SIG.sign}}(d_u^{SK}, $}$r \Vert$\reuse{$ U_3)$.}} \newline
			\textit{Server}: \malicious{\reuse{Collect $\omega'_u$ from at least $t$ users (denote with $U_4 \subseteq U_3$ this set of users). Send to each user in $U_4$ the set $\{v, \omega'_v\}_{v \in U_4}$.}}
			\item \textbf{Stage 4 (Unmasking)}: \newline
			\(
			\textrm{\textit{User} \ } u \left\{
			\begin{minipage}[c]{0.9\linewidth}
				\begin{itemize}
					\item[--] \reuse{Receive from the server a list $\{v$}\malicious{\reuse{$, \omega'_v$}}\reuse{$\}_{v \in U_4}$.  Verify that $U_4 \subseteq U_3$, that $\lvert U_4 \rvert \geq t$} \malicious{\reuse{, that $\textrm{\textbf{SIG.ver}}(d^{PK}_v, $}$r \Vert$\reuse{$ U_3, \omega'_v) = 1$ for all $v \in U_4$, (otherwise abort).}}
					\item[--] For each other user $v \in U_2 \subseteq \{u\}$, decrypt the ciphertext $v' \Vert u' \Vert s_{v, u}^{SK} \Vert b_{v, u}$\xnoise{$\Vert \{g_{v, k, u}\}_{1 \leq k \leq T}$}$\leftarrow \textrm{\textbf{AE.dec}}( $\reuse{$\textrm{\textbf{KA.agree}}(c_u^{SK}, c_v^{PK}), e_{v, u})$} received in the \textbf{MaskedInputCollection} round and \reuse{assert that $u = u'  \wedge v = v'$}.
					\item[--] \reuse{If any of the decryption operations fail (in particular, the ciphertext does not correctly authenticate), abort.}
					\item[--] Send a list of shares to the server: $s_{v, u}^{SK}$ for users $v \in U_2 \setminus U_3$ and $b_{v, u}$ in $v \in U_3$; \xnoise{and a list of seeds $g_{u, k}$ for $\lvert U \setminus U_3 \rvert + 1 \leq k \leq T$}.
				\end{itemize}
			\end{minipage}
			\right.
			\)
			\(
			\textrm{\textit{Server} \ } \left\{
			\begin{minipage}[c]{0.9\linewidth}
				\begin{itemize}
					\item[--] \reuse{Collect responses from at least $t$ users (denote with $U_5 \subseteq U_4$ this set of users).}
					\item[--] For each user $u \in U_2 \setminus U_3$, reconstruct $s_u^{SK} \leftarrow \textrm{\textbf{SS.recon}}(\{s_{u, v}^{SK}\}_{v\in U_5}, t)$ and use it (with the public keys received in the \textbf{AdvertiseKeys} round) to recompute $\bm{p}_{u, v}$ for all $v \in U_3$ using the PRG.
					For each user $u \in U_3$, reconstruct $b_u \leftarrow  \textrm{\textbf{SS.recon}}(\{b_{u, v}\}_{v\in U_5}, t)$ to recompute $\bm{p}_u$ using the PRG.
					\item[--] Compute $\bm{z} = \sum_{u\in U_3} \tilde{\bm{\Delta}}_u$ as $\sum_{u\in U_3} \tilde{\bm{\Delta}}_u = \sum_{u\in U_3} \bm{y_u} - \sum_{u\in U_3} \bm{p_u} + \sum_{u \in U_3, v \in U_2 \setminus U_3} \bm{p}_{v, u}$.
					\item[--] \xnoise{Sends to each user $u \in U_5$ the set $U_5$, if $U_3 \setminus U_5 \neq \emptyset$.}
				\end{itemize}
			\end{minipage}
			\right.
			\)
			\item \xnoise{\textbf{Stage 5 (ExcessiveNoiseRemoval)}}: \newline
			\(
			\textrm{\textit{User} \ } u \left\{
			\begin{minipage}[c]{0.9\linewidth}
				\begin{itemize}
					\item[--] \xnoise{Receive from the server the set $U_5$. Verify that $U_5 \subseteq U_4$ and that $U_5 \geq t$ (otherwise abort).}
					\item[--] \xnoise{Send a list of shares to the server, which consists of $g_{v, k, u}$ for users $v \in U_3 \setminus U_5$ and $\lvert U \setminus U_3 \rvert + 1 \leq k \leq T$.}
				\end{itemize}
			\end{minipage}
			\right.
			\)
			\(
			\textrm{\textit{Server} \ } \left\{
			\begin{minipage}[c]{0.9\linewidth}
				\begin{itemize}
					\item[--] \xnoise{Collect responses from at least $t$ users (denote with $U_6 \subseteq U_5$ this set of users).}
					\item[--] \xnoise{For each user $u \in U_3 \setminus U_5$, reconstruct $g_{u, k} \leftarrow  \textrm{\textbf{SS.recon}}(\{g_{u, k, v}\}_{v\in U_6}, t)$ for $\lvert U \setminus U_3 \rvert + 1 \leq k \leq T$.}
					\item[--] \xnoise{Generate random noises $\bm{n}_{u, k} \stackrel{g_{u, k}}{\sim} \chi \left( \frac{\sigma_*^2}{(\lvert U \rvert - k + 1)(\lvert U \rvert - k )} \cdot \frac{t}{t - T_C} \right)$ for $u \in U_3$ and $\lvert U \setminus U_3 \rvert + 1 \leq k \leq T$ and subtracts them from $\bm{z}$.}
				\end{itemize}
			\end{minipage}
			\right.
			\)
		\end{itemize} \\
		\bottomrule
		\vspace*{-0.2in}
		\captionof{figure}{
			Detailed description of the SecAgg protocol~\cite{bonawitz2017practical} integrated with \texttt{XNoise} (\cref{sec:enforcement_add}).
			\malicious{Italicized parts inside square brackets are required to guarantee security in the malicious threat model.}
			\xnoise{Red, underlined parts are specific for \texttt{XNoise}.}
			\reuse{Green, highlighted parts are secure results of SecAgg reused by \texttt{XNoise}.}
			The symbol $\Vert$ denotes concatenation.
		}
		\label{fig:enforcement_integration}
	\end{tabularx}
\end{table*}

To exemplify, Figure~\ref{fig:enforcement_integration} details how we integrate \texttt{XNoise} with SecAgg~\cite{bonawitz2017practical} for reusing the secure channels across clients and correct broadcast on dropout outcome.
SecAgg is instantiated with a public key infrastructure (PKI), the Diffie-Hellman key agreement ~\cite{merkle1978secure} \textbf{KA} protocol composed with a secure hash function, the Shamir's $t$-out-of-$n$ secret sharing scheme~\cite{shamir1979share} \textbf{SS}, an IND-CPA (Indistinguishability against Chosen-Plaintext Attacks) and INT-CTXT (Integrity of Ciphertext) authenticated encryption scheme \textbf{AE}, a UF-CMA signature scheme \textbf{SIG}, and a secure pseudorandom generator \textbf{PRG}.
For a detailed explanation of the cryptographic primitives employed by SecAgg, we refer the reader to the original paper~\cite{bonawitz2017practical}.
While we opted to implement \texttt{XNoise} by resuing SecAgg's infrastructure for improved complexity and efficiency,  we emphasize that \texttt{XNoise} is self-contained and complementary to secure aggregation protocols.

\PHM{Handling Mild Collusion.}
The server can collude with a subset of clients  under semi-honest and malicious threat models.
In its strongest form, they collude from the very beginning of the protocol and pool their views all the time.
As the collusion scale is presumably mild (justified in~\cref{sec:background_threat}), honest clients can use a slightly increased amount of local noise to handle such collusion without utility loss.
For example, given a collusion tolerance $T_C \approx 0.01\lvert U \rvert$, instead of adding noise of level $\frac{\sigma_*^2}{\lvert U \rvert - T}$, each honest client adds noise of level $\frac{\sigma_*^2}{\lvert U \rvert - T} \cdot \frac{t}{t - T_C}$ where $t$ is the threshold as used in SecAgg.
In doing so, a collusion within the tolerance $T_C$ will not yield an insufficiently perturbed aggregate (\cref{sec:enforcement_security}).

It is important to mention that in the malicious setting with mild collusion, \system{} no longer enforces the minimum necessary noise but instead introduces a noise inflation factor of $\frac{t}{t - T_C}$.
Fortunately, given that $t$ is intentionally much larger than $T_C$,\footnote{The feasible range of $t$ is $(0.5 \lvert U \rvert, \lvert U \rvert]$ in the malicious setting~\cite{bonawitz2017practical}.}, the inflation factor is only slightly greater than $1$.
It should be noted, however, that this approach alone is insufficient to address the privacy leakage caused by dropout without any loss in utility.
We anticipate that dropout could be on a much larger scale than collusion, potentially reaching the same magnitude as $t$ (\cref{sec:issue_client_dropout}).

\subsection{Security Analysis}
\label{sec:enforcement_security}

We consider the strongest adversary in our threat model (\cref{sec:background_threat}), i.e., a malicious server colluding with a subset of sampled clients, as it subsumes weaker adversaries.
By being malicious, we mean arbitrary deviation from the protocol, e.g., sending incorrect and/or chosen messages to honest clients, aborting, or omitting messages.
In this case, we aim to provide the target level of differential privacy for honest clients (i.e., the adversary never sees an insufficiently perturbed update).
We regard the protocol illustrated in Figure~\ref{fig:enforcement_integration} (denoted by $\pi$) as the target implementation of \texttt{XNoise} without loss of generality.

Let $\eta$ be security parameter;
$t$ be the threshold of SecAgg;
$T$ and $T_C$ be the tolerated number of clients dropping out and colluding in \texttt{XNoise}, respectively;
$U$ be the set of sampled clients;
$C \subset U \cup \{S\}$ be the set of colluding parties (where $S$ is the server);
$\bm{\Delta}_{U'} = \{\bm{\Delta}_u\}_{u \in U'}$ and $g_{U'} = \{g_{u, k}\}_{u \in U', k \in [0, T]}$ be the input vectors and sampling seeds used in noise addition of any subset of users $U' \subseteq U$, respectively;
$\sigma_*^2$ be the target level of aggregate noise for each round.
Theorem~\ref{thm:security} shows that a computationally bounded adversary cannot recover an insufficiently perturbed aggregate during the execution of $\pi$.
We give the proof in Appendix~\ref{sec:appendix_security}.

\begin{theorem}[Privacy against Malicious Adversaries]
For all $\eta$, $t$, $U$, $T$, $T_C$, $C \subseteq U \cup \{S\}$, $\bm{\Delta}_{U \setminus C}$ and $g_{U \setminus C}$. If $2t > \lvert U \rvert + \lvert C \cap U \lvert$ and $T_C \geq \lvert C \cap U \rvert$, probabilistic polynomial-time (PPT) adversary, given its view of an execution of $\pi$, cannot recover an aggregate update perturbed with noise less than the target level $\sigma_*^2$ with non-negligible probability.
\label{thm:security}
\end{theorem}
\section{Optimal Pipeline Acceleration}
\label{sec:pipeline}

As described in \cref{sec:secagg_issue}, secure aggregation protocols used for distributed DP in FL create a severe performance bottleneck in the round latency.
Additionally, integrating our dropout-resilient noise enforcement scheme may further exacerbate its inefficiency.
To address this performance issue, we target system-level solutions that preserve all the existing merits of a specific secure aggregation protocol.

\PHM{Technical Intuition.}
We first identified three types of operations in distributed DP that use different
system resources: 1) \texttt{s-comp} that
uses the compute resources (e.g., CPU and memory) of the server,
2) \texttt{c-comp} that uses the compute resources of clients, and
3) \texttt{comm} that relies on server-client communications.
As observed in FL practice, plain execution of distributed DP leads to low utilization of these resources over time: \texttt{s-comp}, \texttt{c-comp}, and \texttt{comm} can be idle for up to 53\%, 63\%, and 93\% of the round time, respectively.
This indicates that pipelining, which enables overlapping resource usage, is viable for improving the utilization.
However, designing and automating pipelined execution for distributed DP in FL leads to two novel system challenges:
\begin{itemize}
	\item Given the large variety in secure aggregation protocols, how to represent their workflows to facilitate generic solutions to the pipelined execution problem (\cref{sec:pipeline_staging})?
	\item {Given the complexity of FL environments, how to correctly model them for generating an optimal pipelining plan for maximum acceleration (\cref{sec:pipeline_determining})?}
\end{itemize}

\subsection{Staging Workflow for Pipelined Execution}
\label{sec:pipeline_staging}

To allow for a generic solution, we first abstract away the specifics of secure aggregation protocols.

\PHM{Abstracting Workflow for Generality.} Unlike traditional ML workflows, which involve a computation graph of neural network with general consensus established on their representation~\cite{abadi2016tensorflow, jia2019taso}, secure aggregation protocols combine both computation and communication, with no standardized approach.
Nonetheless, we notice that secure aggregation protocols are often designed as multi-round server-client interactions.
We thus propose to represent the workflow of a secure aggregation protocol as a sequence of `round-trip \emph{steps}', each of which starts with the server's request for some data and ends with the related clients' responses.
Furthermore, we associate each step with its dominant system resource and group consecutive steps that use the same resource into a \emph{stage} which corresponds to the minimum scheduling unit in pipelining.
For example, Table~\ref{tab:pipeline_stage} illustrates our representation of a distributed DP protocol using SecAgg~\cite{bonawitz2017practical}, where all 11 steps are grouped into 5 stages.

\begin{table}[t]
	\centering
	\caption{Abstracting the workflow of dropout-resilient distributed DP into multiple stages for pipelined execution.}
	\label{tab:pipeline_stage}
	\resizebox{0.9\columnwidth}{!}{%
		\begin{tabular}{clc}
			\toprule
			Step & Operation & Stage (Resource) \\
			\midrule
			1 & Clients encode updates. & \multirow{4}{*}{1 (\texttt{c-comp})} \\ 
			2 & Clients generate security keys. &  \\
			3 & Clients establish shared secrets. &  \\
			4 & Clients mask encoded updates. & \\ \cmidrule{3-3} 
			5 & Clients upload masked updates. & 2 (\texttt{comm}) \\ \cmidrule{3-3}
			6 & Server deals with dropout. & \multirow{3}{*}{3 (\texttt{s-comp})} \\
			7 & Server computes aggregate update. &  \\
			8 & Server updates the global model. &  \\ \cmidrule{3-3}
			9 & Server dispatches the aggregate. & 4 (\texttt{comm}) \\ \cmidrule{3-3}
			10 & Clients decode the aggregate. & \multirow{2}{*}{5 (\texttt{c-comp})} \\
			11 & Clients use the aggregate. & \\
			\bottomrule
		\end{tabular}
	}%
\end{table}

\PHM{Pipelining via Task Partitioning.}
By construction, any two adjacent stages consume different system resources, enabling overlapped execution among independent aggregation workflows.
Leveraging the coordinate-wise nature of aggregation,
\system{} partitions each client's update $\Delta_i$ into $m$ chunks $\Delta_{i,1}, \cdots, \Delta_{i,m}$.
This divides the original aggregation task into $m$ independent sub-tasks, where the $j$-th sub-task aggregates the $j$-th chunks of all clients, i.e.,
$\sum_i \Delta_i = (\sum_i \Delta_{i,1}) \Vert \cdots \Vert (\sum_i \Delta_{i,m})$, 
where $\Vert$ denotes concatenation.
As such, \system{} can enable pipeline parallelism by scheduling the processing stages of the $m$ sub-tasks.

\begin{figure}[t]
	\centering
	\includegraphics[width=1.0\columnwidth]{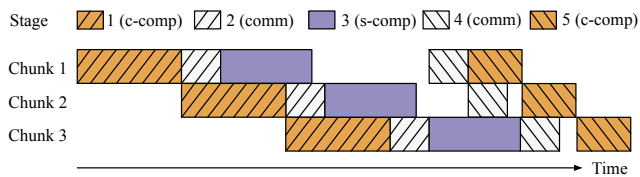}
	\caption{Pipeline scheduling of 3 chunk-aggregation tasks for distributed DP in 5 stages, as specified in Table~\ref{tab:pipeline_stage}.}
	\label{fig:pipeline}
\end{figure}

\begin{figure*}[t]
	\centering
	\includegraphics[width=1.0\linewidth]{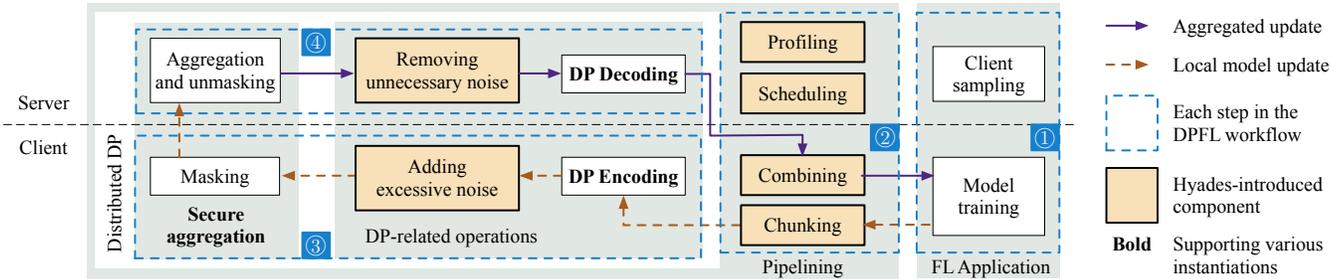}
	\caption{An overview of \system{}' system architecture and how it fits in 
		the existing FL workflow.}
	\label{fig:overview}
\end{figure*}

\PHM{Reducing Design Space.}
Note that each partition configuration is associated with a completion time that is achieved by executing that configuration.
Hence, deriving an optimal pipelined execution plan essentially requires searching for the configuration with the shortest completion time.
To avoid a complex combinatorial problem, we focus on evenly partitioning model updates, which reduces the problem to deciding only $m$, i.e., the number of chunks.
Figure~\ref{fig:pipeline} illustrates
a pipeline execution result with 3 equally-sized aggregation sub-tasks.
In Section~\ref{sec:evaluation_efficiency_pipeline}, we show that such a reduced consideration suffices to gain a remarkable speedup in practice.

\subsection{Determining Optimal Number of Chunks}
\label{sec:pipeline_determining}

While a formal description of the optimization problem for pipelined execution is deferred to Appendix~\ref{sec:appendix_pipeline} for brevity,
it is clear that the key to solving an optimal pipelining plan lies in accurately determining the completion time for each possible choice of $m$.
This requires a performance model and a profiling approach that both fit the FL practice.

\PHM{Performance Model with Intervention Accounted.}
To compute the completion time associated with a specific $m$, \system{} relies on the following performance model that empirically characterizes how the processing time for a sub-task at a stage $s$, denoted by $\tau_s$, is related to $m$:
\begin{equation}
	\tau_s = \beta_{s,1} \frac{d}{m} + \beta_{s,2} {m} + \beta_{s,3},
	\label{eq:pipeline}
\end{equation}
where $d$ is the update size; $\beta_{s,1}, \beta_
{s,2}$ and $\beta_{s,3}$ are the profiled parameters
used to weigh the impact of partition size, inter-task intervention, and constant cost, respectively.
The first and third terms are intuitive, while the second term is specifically designed for FL.
Unlike traditional ML where each node can be dedicated to one task, resources in FL are scarce, less capable, and do not provide strong isolation across tasks.
As the only device contributed by the owner client, a mobile device is seldom fully committed to pipeline stages that use \texttt{c-comp}.
Instead, some of its CPU cycles will be spent on network IO to facilitate stages that use \texttt{comm}.
Such distraction can be accentuated as the pipeline goes deeper.
\system{} thus respects this intervention effect across pipelined tasks.

 \PHM{Parameter Profiling.}
$\beta_1, \beta_2$, and $\beta_3$ in Equation~\ref{eq:pipeline} depend only on the hardware capabilities of the server and the participating clients, especially the slowest one.
As the sampled clients' tail latency does not vary vastly across rounds in practice, we let \system{} profile these constants by linear regression with offline micro-benchmarking, which executes the protocol with small-scale proxy data for certain rounds.
Note that such lightweight profiling can also be conducted online by interleaving it with the training workflow if needed.


\section{Implementation}
\label{sec:implementation}

We have implemented \system{} with 10.3k lines of Python code. It leverages \texttt
{PyTorch}~\cite{paszke2019pytorch} to instantiate FL applications and
employs the distributed DP protocol with DSkellam~\cite
{agarwal2021skellam}.

\PHM{System Architecture.} 
Figure~\ref{fig:overview} shows how \system{} fits in the existing FL workflow, with yellow boxes being \system-introduced components.  \textcircled{\raisebox{-1.0pt}{1}} \emph{Client sampling and training}: at the beginning of each round, the server randomly samples a subset of available clients as participants. The sampled clients then fetch the global model from the server and compute local updates using private data.
\textcircled{\raisebox{-1.0pt}{2}} \emph{Pipeline preparation}: based on the optimal pipeline execution plan provided by the server, each client chunks the local update for pipelined aggregation.
\textcircled{\raisebox{-1.0pt}{3}} \emph{Client processing}: for each update chunk, the client perturbs it with DSkellam's encoding scheme  and our \texttt{XNoise} noise enforcement approach. The perturbed chunk is further masked following the secure aggregation protocol. 
\textcircled{\raisebox{-1.0pt}{4}} \emph{Server aggregation}: the server aggregates and unmasks the received update chunks, removes the excessive parts of their DP noises to ensure that the residual noise remains at the minimum required level, and uses each aggregated update chunk to refine the respective part of the global model.

\PHM{Programming Interface.}
Despite our prototype choice, \system{} is proactively designed to be complementary to existing differentially private FL (DPFL) frameworks.
In particular, \system{} offers a user-friendly programming interface for developers to implement a variety of privacy and security building blocks.
Further details are provided in
Appendix~\ref{sec:appendix_programming}.
To the best of our knowledge, \system{} is the first generic and end-to-end implementation of distributed DP with pipeline acceleration.

\section{Evaluation}
\label{sec:evaluation}

We evaluate \system{}'s effectiveness on three CV and NLP FL tasks in the semi-honest setting.
The highlights of our evaluation are listed below.

\begin{enumerate}[nosep]
    \item Our noise enforcement scheme, \texttt{XNoise}, ensures that the target privacy level is consistently attained, even when client dropout occurs, without impairing model utility (\cref{sec:evaluation_privacy}).
    The runtime overhead is deemed acceptable even without pipeline acceleration, and the network overhead remains constant despite the model's expanding size (\cref{sec:evaluation_efficiency_noise}).
    \item The pipeline-parallel aggregation design in \system{} significantly enhances the training speed, resulting in up to 2.4$\times$ improvement in the round time (\cref{sec:evaluation_efficiency_pipeline}).
\end{enumerate}

\subsection{Methodology}
\label{sec:evaluation_methodology}

\PHB{Datasets and Models.}
We run two categories of applications with three real-world datasets of different scales.

\begin{itemize}
	\item \emph{Image Classification}: the first dataset, CIFAR-10~\cite{krizhevsky2009learning}, consists of 60k colored images categorized into 10 classes.
	We train a ResNet-18~\cite{he2016deep} model with 11M parameters over 100 clients using
	a non-IID data distribution, generated by applying latent Dirichlet allocation (LDA)~\cite{hsu2019measuring, reddi2020adaptive, al2020federated, acar2021federated} with concentration parameters set to 1.0 (i.e., label distributions highly skewed across clients).
	The second dataset, FEMNIST~\cite{caldas2018leaf}, consists of 805k greyscale images classified into 62 classes.
	The dataset was partitioned by the original data owners, and we merge every three owners' data to form a client's dataset.
	We train a CNN model~\cite{kairouz2021distributed, agarwal2021skellam} with 1M parameters over 1000 clients.
	\item \emph{Language Modeling}: the large-scale Reddit dataset~\cite{reddit}.
	We train an Albert~\cite{lan2020albert} model with 15M parameters over 200 clients for next-word prediction.
\end{itemize}

\PHM{Experiment Setup.}
We launch an AWS EC2 \texttt{r5.4xlarge} instance (16 vCPUs and 128 GB memory) for the server and one \texttt{c5.xlarge} (4 vCPUs and 8 GB memory) instance for each client, aiming to match the computing power of mobile devices.
To emulate hardware heterogeneity, we set the response latencies of clients to follow the Zipf distribution~\cite{jiang2022towards, lee2018pretzel, tian2021crystalperf, jia2021boki, jiang2022pisces} with $a = 1.2$ (moderately skewed) such that the end-to-end latency of the $i$-th slowest client is proportional to $i^{-a}$.
We also emulate network heterogeneity by throttling clients' bandwidth to fall within the range [21Mbps, 210Mbps] to match the typical mobile bandwidth~\cite{ciscoannual} and meanwhile follow another independent Zipf distribution with $a = 1.2$.

\begin{figure}[t]
	\centering
	\begin{subfigure}[b]{0.49\columnwidth}
		\centering
		\includegraphics[width=\columnwidth]{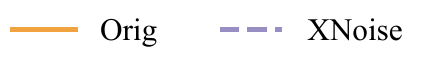}
	\end{subfigure} \\
	\begin{subfigure}[b]{0.32\columnwidth}
		\centering
		\includegraphics[width=\columnwidth]{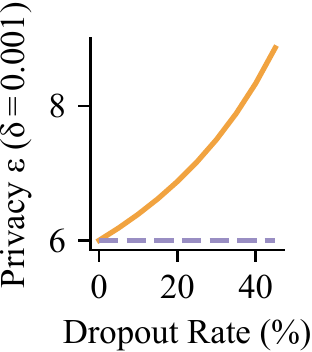}
		\caption{FEMNIST.}
		\label{fig:evaluation_privacy_cnn}
	\end{subfigure} \hfill
	\begin{subfigure}[b]{0.32\columnwidth}
		\centering
		\includegraphics[width=\columnwidth]{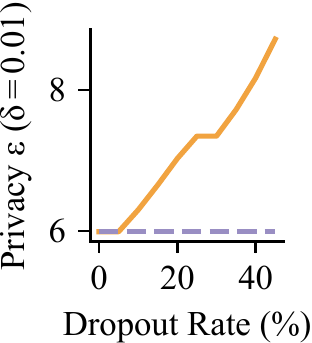}
		\caption{CIFAR-10.}
		\label{fig:evaluation_privacy_resnet18}
	\end{subfigure} \hfill
	\begin{subfigure}[b]{0.32\columnwidth}
		\centering
		\includegraphics[width=\columnwidth]{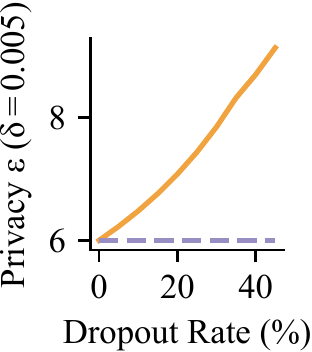}
		\caption{Reddit.}
		\label{fig:evaluation_privacy_albert}
	\end{subfigure}
	\caption{
		Privacy budget consumption.
		A larger $\epsilon$ corresponds to worse privacy preservation.
	}
	\label{fig:evaluation_privacy}
\end{figure}

\PHM{Hyperparameters.}
For FL training, we use the mini-batch SGD for FEMNIST and CNN and AdamW~\cite{loshchilov2018decoupled} for Reddit, all with momentum set to 0.9.
The number of training rounds and local epochs are 50 and 2 for both FEMNIST and Reddit, and 150 and 1 for CIFAR-10, respectively.
The batch size and learning rate are 20 and 0.01 for FEMNIST, 128 and 0.005 for CIFAR-10, and 20 and 8e-5 for Reddit, respectively.

For distributed DP, we set the privacy budget $\epsilon$ to $6$ and $\delta$ the reciprocal of the total number of clients as they represent standard privacy in the DPFL literature~\cite{agarwal2021skellam, kairouz2021distributed, stevens2022efficient}.
We fix the signal bound multiplier $k=3$, bias $\beta = e^{-0.5}$, and bit-width $b=20$ for the configuration of DSkellam~\cite{agarwal2021skellam} as specified in the original paper.
The L2-norm clipping bounds~\cite{agarwal2021skellam} for FEMNIST and CIFAR-10 is set to 1 and 3, respectively. Also, there are 100 and 16 clients being sampled in each round for FEMNIST and CIFAR-10, respectively.
While we base the implementation of secure aggregation on SecAgg~\cite{bonawitz2017practical} when evaluating our noise enforcement approach, \texttt{XNoise}, we also implement SecAgg+~\cite{bell2020secure} in Section~\ref{sec:evaluation_efficiency_pipeline} to highlight the generality of our distributed pipeline architecture.

\PHM{Dropout Model.}
We assume that when clients drop out of the protocol, they drop out after being sampled but before sending their masked and perturbed update to the server.
To study the impact of various severities, we let clients randomly drop with a configurable rate in each training round.
The dropout rate is consistent within a training process, while varying from 0 to 40\% across different processes.

\PHM{Baseline.} We compare \system{} against \texttt{Orig}, the original, commonly-used distributed DP protocol that lacks the ability to enforce the target noise level (Definition~\ref{def:orig} in \cref{sec:enforcement}) and support pipeline execution.

\subsection{Effectiveness of Noise Enforcement}
\label{sec:evaluation_privacy}

\PHM{\texttt{XNoise} Ensures Privacy without Sacrificing Utility.}
Figure~\ref{fig:evaluation_privacy} displays the end-to-end privacy budget consumption.
As expected, \texttt{XNoise} achieves the target privacy  ($\epsilon = 6$) in all cases by accurately enforcing the target noise (Theorem~\ref{thm:xnoise-prec}).
On the other hand, due to the missing noise contributions from dropped clients, the overall privacy budget consumed by \texttt{Orig} dramatically grows as the severity of client dropout increases.
For example, when the dropout rate is 40\%, training FEMNIST, CIFAR-10, and Reddit to the preset number of rounds ends up consuming an $\epsilon$ of 8.3, 8.2, and 8.7, respectively.

\begin{table}[]
	\centering
	\caption{Final testing accuracy (for \textbf{F}EMNIST and \textbf{C}IFAR-10) or perplexity (for \textbf{R}eddit, the lower, the better) of \texttt{\textbf{Ori}g} and \texttt{\textbf{XNo}ise} across various dropout rates $d$.}
	\label{tab:utility}
	\resizebox{1.0\columnwidth}{!}{%
		\begin{tabular}{c|cc|cc|cc|cc|cc}
			\hline
			$d$ &
			\multicolumn{2}{c|}{0} &
			\multicolumn{2}{c|}{10\%} &
			\multicolumn{2}{c|}{20\%} &
			\multicolumn{2}{c|}{30\%} &
			\multicolumn{2}{c}{40\%} \\ \hline
			&
			\texttt{Ori} &
			\texttt{XNo} &
			\texttt{Ori} &
			\texttt{XNo} &
			\texttt{Ori} &
			\texttt{XNo} &
			\texttt{Ori} &
			\texttt{XNo} &
			\texttt{Ori} &
			\texttt{XNo} \\ \hline
			F  & 61.3 & 61.4 & 61.4 & 61.4 & 61.2 & 61.4 & 61.2 & 61.2  & 61.4 & 61.5 \\
			C & 66.5 & 66.3 & 66.7  & 66.9 & 66.6 & 65.7 & 64.3 & 65.7 & 63.8 & 64.2  \\
			R & 2169 & 2142 & 2158  & 2179 & 2286 & 2285 & 2294 & 2317 & 2299 & 2329  \\
			\hline
		\end{tabular}
	}%
\end{table}

\begin{figure}[t]
	\centering
	\begin{subfigure}[b]{0.49\columnwidth}
			\centering
			\includegraphics[width=\columnwidth]{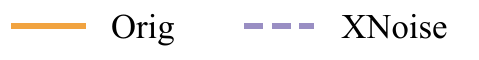}
		\end{subfigure} \\
	\begin{subfigure}[b]{0.32\columnwidth}
			\centering
			\includegraphics[width=\columnwidth]{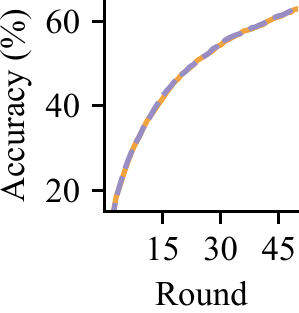}
			\caption{FEMNIST.}
			\label{fig:evaluation_round_acc_cnn}
		\end{subfigure} \hfill
	\begin{subfigure}[b]{0.35\columnwidth}
			\centering
			\includegraphics[width=\columnwidth]{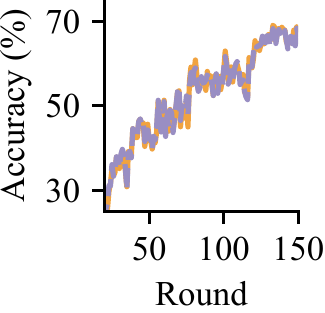}
			\caption{CIFAR-10.}
			\label{fig:evaluation_round_acc_resnet18}
		\end{subfigure} \hfill
		\begin{subfigure}[b]{0.29\columnwidth}
			\centering
			\includegraphics[width=\columnwidth]{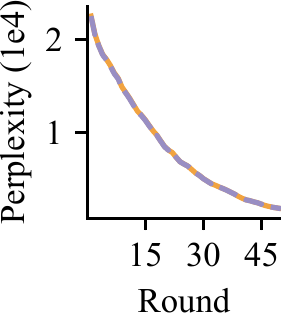}
			\caption{Reddit.}
			\label{fig:evaluation_round_acc_albert}
		\end{subfigure}
	\caption{Round-to-accuracy performance (20\% dropout).}
	\label{fig:evaluation_round_acc}
\end{figure}

\begin{figure*}[t]
    \centering
    \begin{subfigure}[b]{0.45\linewidth}
        \centering
        \includegraphics[width=\columnwidth]{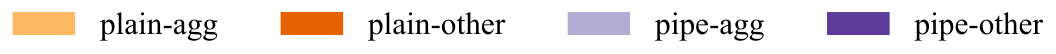}
    \end{subfigure}
    \\
    \begin{subfigure}[b]{0.245\linewidth}
        \centering
        \includegraphics[width=\columnwidth]{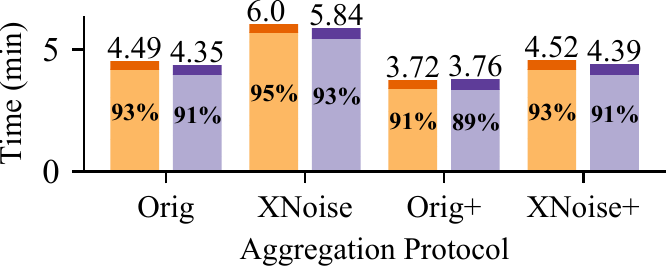}
        \caption{FEMNIST, CNN, $d=0\%$.}
        \label{fig:evaluation_efficiency_femnist_cnn_0}
    \end{subfigure} \hfill
	\begin{subfigure}[b]{0.245\linewidth}
		\centering
		\includegraphics[width=\columnwidth]{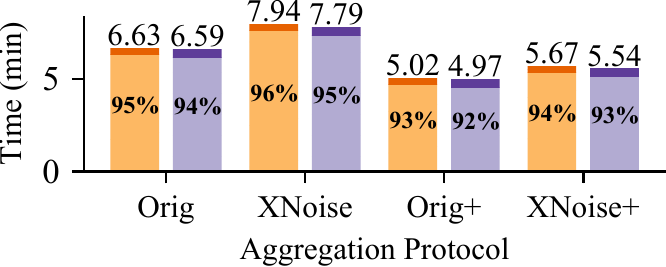}
		\caption{FEMNIST, CNN, $d=10\%$.}
		\label{fig:evaluation_efficiency_femnist_cnn_10}
	\end{subfigure} \hfill
	\begin{subfigure}[b]{0.245\linewidth}
		\centering
		\includegraphics[width=\columnwidth]{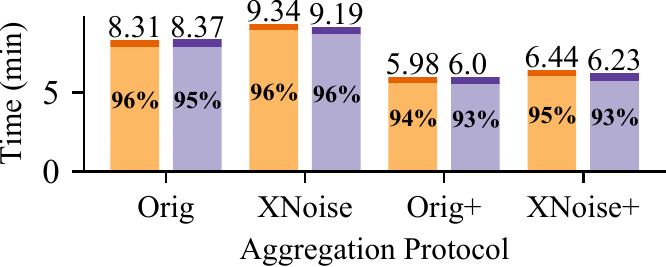}
		\caption{FEMNIST, CNN, $d=20\%$.}
		\label{fig:evaluation_efficiency_femnist_cnn_20}
	\end{subfigure} \hfill
	\begin{subfigure}[b]{0.245\linewidth}
		\centering
		\includegraphics[width=\columnwidth]{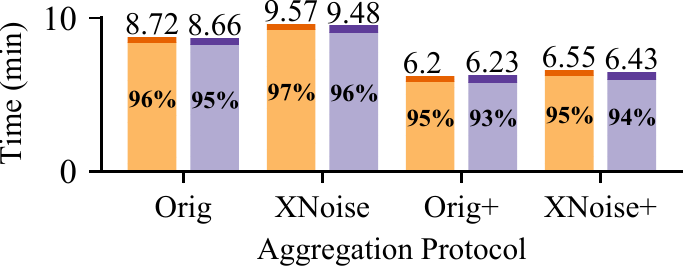}
		\caption{FEMNIST, CNN, $d=30\%$.}
		\label{fig:evaluation_efficiency_femnist_cnn_30}
	\end{subfigure} \hfill
	\begin{subfigure}[b]{0.245\linewidth}
        \includegraphics[width=\columnwidth]{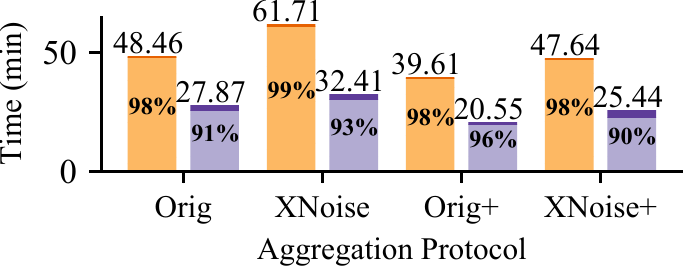}
	\caption{FEMNIST, ResNet-18, $d=0\%$.}
	\label{fig:evaluation_efficiency_femnist_resnet18_0}
	\end{subfigure} \hfill
	\begin{subfigure}[b]{0.245\linewidth}
	\centering
	\includegraphics[width=\columnwidth]{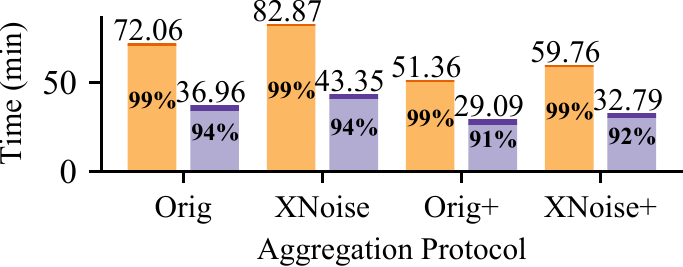}
	\caption{FEMNIST, ResNet-18, $d=10\%$.}
	\label{fig:evaluation_efficiency_femnist_resnet18_10}
	\end{subfigure} \hfill
	\begin{subfigure}[b]{0.245\linewidth}
	\centering
	\includegraphics[width=\columnwidth]{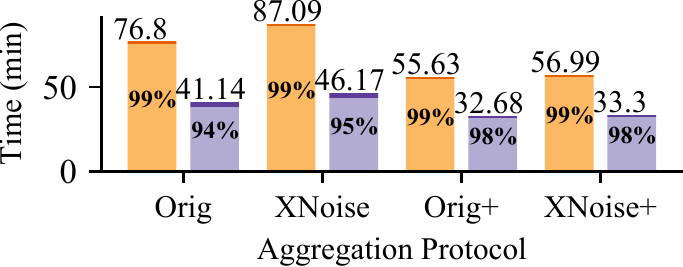}
	\caption{FEMNIST, ResNet-18, $d=20\%$.}
	\label{fig:evaluation_efficiency_femnist_resnet18_20}
	\end{subfigure} \hfill
	\begin{subfigure}[b]{0.245\linewidth}
	\centering
	\includegraphics[width=\columnwidth]{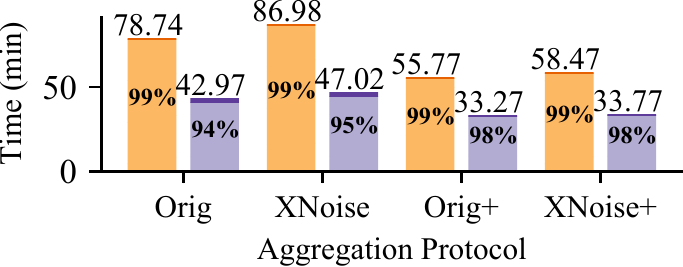}
	\caption{FEMNIST, ResNet-18, $d=30\%$.}
	\label{fig:evaluation_efficiency_femnist_resnet18_30}
	\end{subfigure} \hfill
    \begin{subfigure}[b]{0.245\linewidth}
        \centering
        \includegraphics[width=\columnwidth]{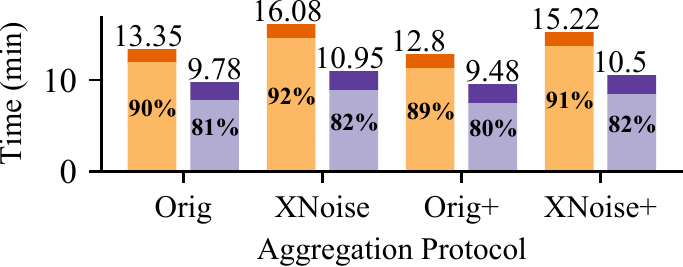}
        \caption{CIFAR-10, ResNet-18, $d=0\%$.}
        \label{fig:evaluation_efficiency_cifar10_resnet18_0}
    \end{subfigure} \hfill
	 \begin{subfigure}[b]{0.245\linewidth}
		\centering
		\includegraphics[width=\columnwidth]{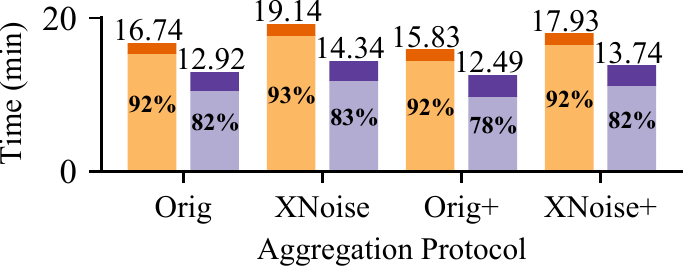}
		\caption{CIFAR-10, ResNet-18, $d=10\%$.}
		\label{fig:evaluation_efficiency_cifar10_resnet18_10}
	\end{subfigure} \hfill
    \begin{subfigure}[b]{0.245\linewidth}
		\centering
		\includegraphics[width=\columnwidth]{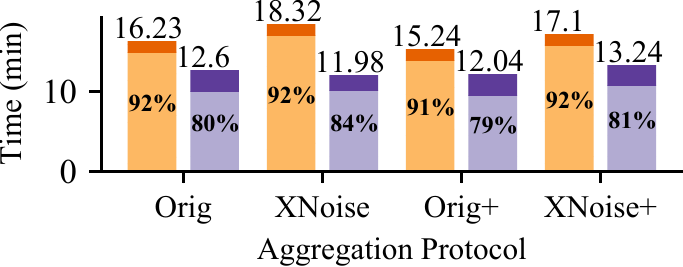}
		\caption{CIFAR-10, ResNet-18, $d=20\%$.}
		\label{fig:evaluation_efficiency_cifar10_resnet18_20}
	\end{subfigure} \hfill
	 \begin{subfigure}[b]{0.245\linewidth}
		\centering
		\includegraphics[width=\columnwidth]{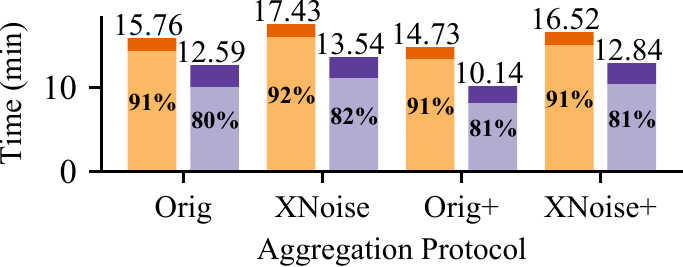}
		\caption{CIFAR-10, ResNet-18, $d=30\%$.}
		\label{fig:evaluation_efficiency_cifar10_resnet18_30}
	\end{subfigure} \hfill
    \begin{subfigure}[b]{0.245\linewidth}
        \centering
        \includegraphics[width=\columnwidth]{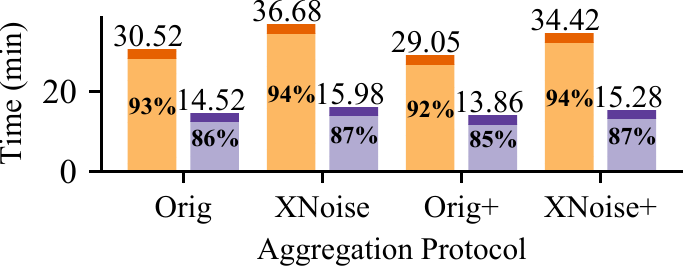}
        \caption{CIFAR-10, VGG-19, $d=0\%$.}
        \label{fig:evaluation_efficiency_cifar10_vgg19_0}
    \end{subfigure}\hfill
	\begin{subfigure}[b]{0.245\linewidth}
		\centering
		\includegraphics[width=\columnwidth]{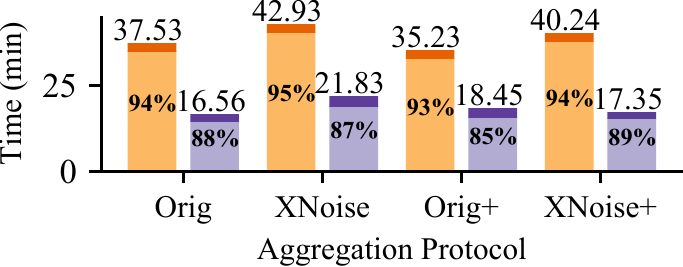}
		\caption{CIFAR-10, VGG-19, $d=10\%$.}
		\label{fig:evaluation_efficiency_cifar10_vgg19_10}
	\end{subfigure} \hfill
    \begin{subfigure}[b]{0.245\linewidth}
        \centering
        \includegraphics[width=\columnwidth]{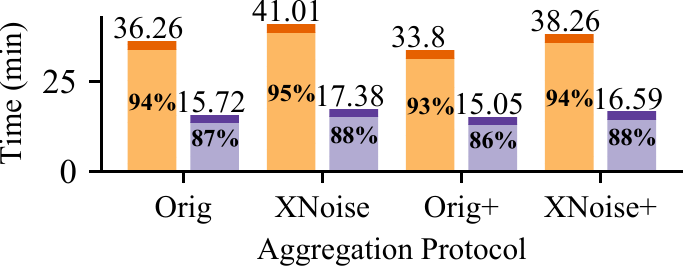}
        \caption{CIFAR-10, VGG-19, $d=20\%$.}
        \label{fig:evaluation_efficiency_cifar10_vgg19_20}
    \end{subfigure} \hfill
	\begin{subfigure}[b]{0.245\linewidth}
		\centering
		\includegraphics[width=\columnwidth]{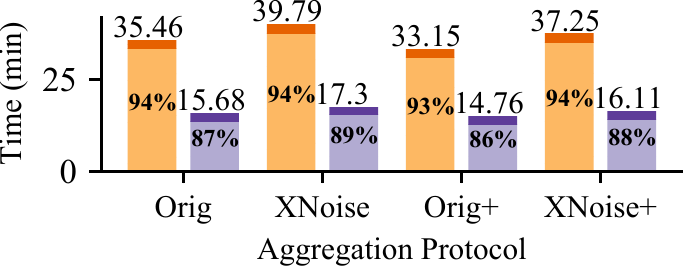}
		\caption{CIFAR-10, VGG-19, $d=30\%$.}
		\label{fig:evaluation_efficiency_cifar10_vgg19_30}
	\end{subfigure} \hfill
    \caption{The round time breakdown for \textbf{plain} execution and \textbf{pipe}line acceleration. The implemented secure aggregation is SecAgg (w/o \textbf{+}) and SecAgg+ (w/ \textbf{+}). The number of sampled clients for FEMNIST and CIFAR-10 are 100 and 16, respectively. The model size for CNN, ResNet-18, and VGG-19 is 1M, 11M, and 20M, respectively. $d$ is the per-round dropout rate.}
    \label{fig:evaluation_efficiency}
\end{figure*}

We also report the final model accuracy in Table~\ref{tab:utility}.
Compared to \texttt{Orig}, which fails to achieve the target privacy in the presence
of client dropout, our \texttt{XNoise} converges at the same speed 
(as exemplified by the learning curves in Figure~\ref{fig:evaluation_round_acc} where the per-round dropout rate is 20\%)
and induces no more than 0.9\% accuracy loss as it uses the minimum noise required to maintain the desired level of privacy.
It should be noted that the accuracies obtained with \system{} are similar to those in other FL studies that use distributed DP~\cite{kairouz2021distributed, agarwal2021skellam, stevens2022efficient}.
This is because FL clients have limited and non-IID data, and model updates are discretized for secure aggregation.
Additionally, in some cases, \texttt{XNoise} achieves higher accuracy than \texttt{Orig}, even though the latter uses less noise.
This is because the stochasticity introduced by slight additional random noise can act as a regularizer to reduce overfitting~\cite{neelakantan2015adding, liu2020fedsel}.

\subsection{Efficiency of Noise Enforcement}~\label{sec:evaluation_efficiency_noise}

\PHM{\texttt{XNoise} Induces Acceptable Overhead in Time.}
Figure~\ref{fig:evaluation_efficiency} shows the average round time, broken down into two components: `agg', related to distributed DP operations, and `other', related to remaining workflow operations such as model training.
In a plain execution without pipeline acceleration (yellow bars for \texttt{Orig} and \texttt{XNoise} in Figure~\ref{fig:evaluation_efficiency_femnist_cnn_0} to~\ref{fig:evaluation_efficiency_femnist_cnn_30} and Figure~\ref{fig:evaluation_efficiency_cifar10_resnet18_0} to~\ref{fig:evaluation_efficiency_cifar10_resnet18_30}), \texttt{XNoise} extends the round time
by up to 34\% given no dropout, and by up to 19\%, 13\%, and 12\% when the per-round dropout rate is 10\%, 20\%, and 30\%, respectively.
This implies a negative relationship between the time cost and dropout severity, as the more clients that drop out, the less noise the server needs to remove in \texttt{XNoise} (Definition~\ref{def:xnoise-prec}).
Such a relationship also implies an acceptable overhead in practice, as \system{} targets scenarios with client dropout.
Moreover, the overhead can be further reduced by \system{}'s pipeline acceleration (\cref{sec:evaluation_efficiency_pipeline}).

\begin{table}[]
	\centering
	\caption{Additional per-round network footprint in MB for a surviving client in `\textbf{r}ebasing' and \texttt{\textbf{X}Noise}, compared to \texttt{Orig}, across various dropout rate $d$.}
	\label{tab:network}
	\resizebox{1.0\columnwidth}{!}{%
		\begin{tabular}{cc|cc|cc|cc}
			\hline
			\multicolumn{2}{c|}{Model size (\# parameters)}                  & \multicolumn{2}{c|}{5M}             & \multicolumn{2}{c|}{50M}            & \multicolumn{2}{c}{500M}            \\ \hline
			\multicolumn{1}{c|}{$d$}                     & \# sampled clients   & \multicolumn{1}{c|}{r} & X           & \multicolumn{1}{c|}{r} & X           & \multicolumn{1}{c|}{r} & X           \\ \hline
			\multicolumn{1}{c|}{\multirow{3}{*}{0\%}} &
			100 &
			\multicolumn{1}{c|}{\multirow{12}{*}{11.9}} &
			0.6 &
			\multicolumn{1}{c|}{\multirow{12}{*}{119.2}} &
			0.6 &
			\multicolumn{1}{c|}{\multirow{12}{*}{1192.1}} &
			0.6 \\
			\multicolumn{1}{c|}{}                      & 200 & \multicolumn{1}{c|}{} & 2.4 & \multicolumn{1}{c|}{} & 2.4 & \multicolumn{1}{c|}{} & 2.4 \\
			\multicolumn{1}{c|}{}                      & 300 & \multicolumn{1}{c|}{} & 5.5 & \multicolumn{1}{c|}{} & 5.5 & \multicolumn{1}{c|}{} & 5.5 \\ \cline{1-2} \cline{4-4} \cline{6-6} \cline{8-8} 
			\multicolumn{1}{c|}{\multirow{3}{*}{10\%}} & 100 & \multicolumn{1}{c|}{} & 0.6 & \multicolumn{1}{c|}{} & 0.6 & \multicolumn{1}{c|}{} & 0.6 \\
			\multicolumn{1}{c|}{}                      & 200 & \multicolumn{1}{c|}{} & 2.4 & \multicolumn{1}{c|}{} & 2.4 & \multicolumn{1}{c|}{} & 2.4 \\
			\multicolumn{1}{c|}{}                      & 300 & \multicolumn{1}{c|}{} & 5.3 & \multicolumn{1}{c|}{} & 5.3 & \multicolumn{1}{c|}{} & 5.3 \\ \cline{1-2} \cline{4-4} \cline{6-6} \cline{8-8} 
			\multicolumn{1}{c|}{\multirow{3}{*}{20\%}} & 100 & \multicolumn{1}{c|}{} & 0.6 & \multicolumn{1}{c|}{} & 0.6 & \multicolumn{1}{c|}{} & 0.6 \\
			\multicolumn{1}{c|}{}                      & 200 & \multicolumn{1}{c|}{} & 2.3 & \multicolumn{1}{c|}{} & 2.3 & \multicolumn{1}{c|}{} & 2.3 \\
			\multicolumn{1}{c|}{}                      & 300 & \multicolumn{1}{c|}{} & 5.2 & \multicolumn{1}{c|}{} & 5.2 & \multicolumn{1}{c|}{} & 5.2 \\ \cline{1-2} \cline{4-4} \cline{6-6} \cline{8-8} 
			\multicolumn{1}{c|}{\multirow{3}{*}{30\%}} & 100 & \multicolumn{1}{c|}{} & 0.6 & \multicolumn{1}{c|}{} & 0.6 & \multicolumn{1}{c|}{} & 0.6 \\
			\multicolumn{1}{c|}{}                      & 200 & \multicolumn{1}{c|}{} & 2.3 & \multicolumn{1}{c|}{} & 2.3 & \multicolumn{1}{c|}{} & 2.3 \\
			\multicolumn{1}{c|}{}                      & 300 & \multicolumn{1}{c|}{} & 5.2 & \multicolumn{1}{c|}{} & 5.2 & \multicolumn{1}{c|}{} & 5.2 \\ \hline
		\end{tabular}
	}%
\end{table}

\PHM{\texttt{XNoise}'s Network Overhead is Invariant of Model Size.}
The noise decomposition in \texttt{XNoise} allows for the transmission of noise seeds (\cref{sec:enforcement_intuition}).
To examine its practical advantage over `rebasing', we benchmark their additional network footprint induced to a surviving client compared to \texttt{Orig}.
As specified in real deployments, the size of a model weight, noise seed, Shamir share of seed, ciphertext of a share (the latter three are configured for \texttt{XNoise} only) are set to 2.5, 32, 16, and 120 in bytes, respectively.
As seen in Table~\ref{tab:network}, as the model size increases, the network overhead of \texttt{XNoise} remains constant and low, while that of `rebasing' grows linearly.
The communication time of `rebasing' can thus be prohibitive when the model is large in size (e.g., >500M) and/or the client's bandwidth is low (e.g., <1Mbps).
Additionally, a large overhead makes client failures in the middle of noise removal more likely, which `rebasing' cannot handle without the loss of model utility (\cref{sec:enforcement_intuition}).

It is important to highlight that the cost of \system{} is dependent on the number of sampled clients rather than the overall population size.
In FL systems, although the population can be extensive, the number of sampled clients is typically restricted to a few hundred~\cite{kairouz2019advances, bonawitz2019towards}.
This limitation is due to the fact that involving additional clients beyond a certain threshold only yields marginal benefits in terms of accelerating convergence~\cite{mcmahan2017communication, wang2021field}.
Therefore, the results presented in Table~\ref{tab:network} are evaluated from practical settings of sample sizes and are applicable in real-world scenarios.

\subsection{Efficiency of Pipeline Acceleration}~\label{sec:evaluation_efficiency_pipeline}

To study the impact of model sizes, we additionally train a ResNet-18 over the FEMNIST dataset, and a VGG-19 model~\cite{simonyan2014very} (20M parameters) over the CIFAR-10 dataset.

\PHM{\system{} Generally Benefits from Pipelining.}
Figure~\ref{fig:evaluation_efficiency} shows that \system{}'s pipeline acceleration benefits all the evaluated aggregation protocols by providing a generic pipeline architecture (\cref{sec:implementation}).
Specifically, utilizing idle resources with pipelined execution can speed up \texttt{Orig} by up to 2.3$\times$ (resp. 2.2$\times$) when the implemented secure aggregation is SecAgg (resp. SecAgg+), while \texttt{XNoise} can achieve a comparable maximum speedup of 2.4$\times$ (resp. 2.3$\times$) with pipelining.
We also notice that the speedup is similar across different evaluated dropout rates when fixing the use of a protocol.

\PHM{\system{} Gains More Speedup with Larger Models.}
The results in Figure~\ref{fig:evaluation_efficiency_cifar10_resnet18_0} to~\ref{fig:evaluation_efficiency_cifar10_vgg19_30} show that CIFAR-10 with ResNet-18 (11M) is accelerated by 1.3-1.5$\times$, while CIFAR-10 with VGG-19 (20M) is accelerated by 1.9-2.5$\times$ by pipelining, indicating that larger models benefit more from this approach.
Similar observations can be made when comparing the results of FEMNIST with CNN (1M) and with ResNet-18 (11M).
This trend can be explained by Amdahl's law~\cite{amdahl1967validity}: as the aggregation time of larger models has a higher dominant factor $p$ in the round latency (e.g., $p=$ 89-93\% in CIFAR-10 with ResNet-18 and $p=$ 93-95\% in CIFAR-10 with VGG-19), assuming the same speedup $s$ of the aggregation offered by pipelining, one can expect a higher overall speedup $S$ for larger models as $S = 1 / ((1 - p) + p/s)$.

\PHM{\system{} Scales with Number of Sampled Clients.}
As reported by Figure~\ref{fig:evaluation_efficiency_femnist_resnet18_0} to~\ref{fig:evaluation_efficiency_cifar10_resnet18_30}, CIFAR-10 with ResNet-18 (16 sampled clients) and FEMNIST with ResNet-18 (100 sampled clients) both benefit from pipelined execution, with FEMNIST achieving a larger speedup (1.7-2.0$\times$) than CIFAR-10 (1.3-1.5$\times$).
It is important to note that the dominant part of the round time is the aggregation process (\texttt{plain-agg}), which is independent of the learning task.
Therefore, the superior performance of \system{} on FEMNIST with ResNet-18 over CIFAR-10 with ResNet-18 can be mainly attributed to the larger number of sampled clients, indicating that \system{} scales well and may even perform better with large-scale training.
This further suggests that the reason why FEMNIST with CNN gains negligible speedup (Figure~\ref{fig:evaluation_efficiency_femnist_cnn_0} to~\ref{fig:evaluation_efficiency_femnist_cnn_30}) is due to its small model size, not large participation scale.
\section{Discussion}
\label{sec:discussion}

\PHB{Random Client Sampling with VRFs.}
As mentioned in~\cref{sec:background_threat}, we can enforce mild collusion among sampled clients even in the presence of a malicious server through the use of VRFs.
Our key insight is that incorporating verifiable randomness can prevent the server from manipulating the sampling process to include a disproportionate number of dishonest clients.
Given that dishonest clients only represent a small portion of the entire population (e.g., billions of Apple devices~\cite{apple2023report}), a VRF-based client selection can ensure that dishonest clients remain a minority in the sampled participants, thus effectively preventing Sybil attacks.

The expected low base rate of dishonest clients results from the prohibitive costs associated with creating and managing a large number of simulated clients.
The adversary is hindered by the high expenses of registering numerous identities to a public key infrastructure (PKI) operated by a qualified trust service provider~\cite{bonawitz2017practical, bell2020secure}.
Furthermore, maintaining a client botnet at scale incurs substantial monetary expenditure~\cite{zhao2023secure}.

Regarding the use of verifiable randomness, we briefly describe a possible design as introduced in a recent work~\cite{jiang2023secure}.
Initially, the server initiates a training round by making an announcement.
Each client within the population employs a VRF using its private key to generate a random number along with an associated proof, with the current round index serving as input.
By comparing the generated random number with a predetermined threshold, a client determines whether it should participate in the ongoing round.
Once a client decides to join, it informs the server about its random number and provides the corresponding proof.
Upon receiving responses from all participating clients, the server considers them as participants and broadcasts their responses for mutual verification. 
Ultimately, a participant proceeds with the training only if all verification tests are successfully passed.

In the above design, a client does not need to know every other clients to reproduce the sampling process for the entire population, as the server does. Instead, a sampled client only needs to refer to the identities of other sampled clients and verify their randomness. 
Moreover, the server can achieve a fixed sample size by first slightly adjusting the selection threshold for over-selection, and then discarding excessive clients based on indiscriminate criteria on their randomness.

\section{Related Work}
\label{sec:related}

The topic of differentially private FL (DPFL) has seen a surge of interest in recent years.
DP-FedAvg~\cite{mcmahan2018learning, ramaswamy2020training}, Distributed DP-SGD~\cite{balle2020privacy}, and DP-FTRL~\cite
{kairouz2021practical}  initiate DPFL in the central DP model where the server is trusted.
Unlike these works, \system{} studies semi-honest and malicious adversaries.

In distributed DP, recent efforts have made significant progress in integrating secure aggregation with DP mechanisms.
DDGauss~\cite{kairouz2021distributed} and DSkellam~\cite{agarwal2021skellam} provide end-to-end privacy analysis by combining the DP noise addition with SecAgg.
FLDP~\cite{stevens2022efficient} explores aggregation based on the learning with errors (LWE) problem~\cite{regev2009lattices}, using residual errors as DP noise.
\system{} complements these works by providing dropout
resilience and improved execution efficiency.

The work closest to \system{} is~\cite{baek2021enhancing} which tackles the dropout resilience problem in distributed DP.
Unlike \system{}, they adopt a `rebasing' design that is incompatible with transmitting seeds (thus, poor efficiency) and does not handle client dropout during noise removal (thus, poor robustness) (\cref{sec:enforcement_intuition}).
Also, it is unclear how to implement their noise enforcement securely against malicious adversaries.

Replacing secure aggregation with other cryptographic primitives in distributed DP is far less studied in FL mainly due to their inefficiency.
For instance, \cite{truex2019hybrid} explores the use of homomorphic encryption schemes, but induces an overhead of over 16 minutes to encrypt a model of size 1M.

Finally, several studies have focused on enhancing the SecAgg protocol, both in general-purpose scenarios~\cite{bell2020secure, so2021turbo, kadhe2020fastsecagg, so2022lightsecagg} and in those specific to FL~\cite{rathee2022elsa, ma2023flamingo, yang2023fast, liu2023dhsa}, without integration with DP.
\system{} is thus independent of these studies in terms of noise enforcement. \system{} also offers the first generic pipeline solution that can expedite them.

\section{Conclusion}
\label{sec:conclusion}

This paper presents \system{}, an efficient FL framework that
enables efficient and dropout-resilient distributed DP in realistic scenarios.
To handle client dropout, 
\system{} designs an efficient and secure ``add-then-remove'' approach to enforce the required noise at the target level precisely.
\system{} also enables pipeline parallelism for accelerated secure aggregation with a generic distributed parallel architecture.
Compared to existing distributed DP mechanisms in FL, \system{} enforces privacy guarantees with optimal model utility in the presence of client dropout without requiring manual strategies.
It also achieves up to 2.4$\times$ faster training completion.
\begin{acks}
    We extend our special thanks to Peter Kairouz, Zheng Xu, Marco Canini, and Suhaib A. Fahmy for their valuable discussions and suggestions that helped improve the early shape of this work.
    We are also grateful to our shepherd, Manuel Costa, and the anonymous EuroSys reviewers for their constructive feedback, which greatly enhanced the quality of this paper.
    We would like to acknowledge Shaohuai Shi for providing GPU clusters and environment settings.
    This work was supported in part by RGC RIF grant R6021-20 and RGC GRF grant 16211123.
\end{acks}

\bibliographystyle{ACM-Reference-Format}
\bibliography{./main.bib}

\appendix
\section{Proofs for Noise Enforcement Results}~\label{sec:appendix_xnoise}

\PHB{Proof for Theorem~\ref{thm:xnoise-prec} (\cref{sec:enforcement_add}).}

\begin{proof}
Before the server performs noise deduction, the aggregated update is randomized with noise level:

\begin{align*}
    & \sum_{c_j \in U \setminus D} \biggl( \frac{\sigma^2_*}{\lvert U \rvert} + \sum_{k=1}^T \frac{\sigma^2_*}{(\lvert U \rvert - k + 1)(\lvert U \rvert - k)} \biggr) \\
    =& \sum_{c_j \in U \setminus D} \biggl( \frac{\sigma^2_*}{\lvert U \rvert} + \sum_{k=1}^T \Bigl( \frac{\sigma^2_*}{\lvert U \rvert - k} - \frac{\sigma^2_*}{\lvert U \rvert - k + 1} \Bigr) \biggr) \\
    =& \sigma^2_* (\lvert U \rvert - \lvert D \rvert) \Bigl( \frac{1}{\lvert U \rvert} + \frac{1}{\lvert U \rvert - T} - \frac{1}{\lvert U \rvert} \Bigr) = \sigma^2_* \frac{\lvert U \rvert - \lvert D \rvert}{\lvert U \rvert - T}.
\end{align*}

Moreover, the total noise level deducted by the server is:

\begin{align*}
    & \sum_{c_j \in U \setminus D} \sum_{k = \lvert D \rvert + 1}^T \frac{\sigma^2_*}{(\lvert U \rvert - k + 1)(\lvert U \rvert - k)} \\
    =& (\lvert U \rvert - \lvert D \rvert) \sigma^2_* \sum_{k = \lvert D \rvert + 1}^T \Bigl(\frac{1}{\lvert U \rvert - k} - \frac{1}{\lvert U \rvert - k + 1}\Bigr). \\
    =& (\lvert U \rvert - \lvert D \rvert) \sigma^2_* \Bigl( \frac{1}{\lvert U \rvert - T} - \frac{1}{\lvert U \rvert - \lvert D \rvert} \Bigr) = \sigma^2_* \frac{T - \lvert D \rvert}{\lvert U \rvert - T}.
\end{align*}

Thus, the remaining noise level at the aggregated update is $\sigma^2_* \frac{\lvert U \rvert - \lvert D \rvert}{\lvert U \rvert - T} - \sigma^2_* \frac{T - \lvert D \rvert}{\lvert U \rvert - T} = \sigma^2_*$.
\end{proof}

\section{Proofs for Security Results}~\label{sec:appendix_security}

We hereby give formal arguments showing that our integration of \texttt{XNoise} and SecAgg (denoted by $\pi$ and specified in~\cref{sec:enforcement_consolidation}), preserves privacy against malicious adversaries.
As is standard, we consider only computationally-bounded adversarial parties, namely those whose strategies can be described by some probabilistic polynomial-time (PPT) algorithm $M$.
The proof will be performed in the so-called `Random Oracle' model~\cite{bellare1993random}.
We assume that a common random oracle $O$ is available to all the parties, who can each make arbitrarily many oracle queries to $O$ during the course of their execution.
In its primary form, given input $x$ and bit length $l$, $O$ outputs a binary string of length $l$, $O(x)$, such that i) each $O(x)$ is uniformly random and independent, and that ii) repeated queries on the same $x$ and $l$ give the same result.
Additionally, given input $x$, bit length $l$, vector length $V$ and distribution $\chi$, $O$ outputs a list of $L$ bit strings of length $l$, $O_\chi(x)$,  such that i) all the bit strings of $O_\chi(x)$ follow the same distribution $\chi$ and are mutually independent and all $O_\chi(x)$'s are mutually independent, and repeated queries on the same $x$, $l$, $V$, and $\chi$ give the same result.
We assume that all honest parties will substitute all \textbf{PRG} calls and noise sampling with calls to $O$.

Denote with $\bm{\Delta}_{U'} = \{\bm{\Delta}_u\}_{u \in U'}$ and $g_{U'} = \{g_{u, k}\}_{u \in U', k \in [0, T]}$ the inputs of any subset of users $U' \subseteq U$.
For fixed $\eta$, $t$, $U$, $T$, and a set $C$ of corrupted parties, we let $M_C$ indicate the polynomial-time algorithm that denotes the `next-message' function of parties in $C$.
We allow $M_C$ to dynamically choose which users to abort in each round, rather than the aborts being statically fixed beforehand.

The following lemma shows that the joint view of colluding corrupt parties in a real execution of $\pi$ can be simulated by \texttt{SIM} given only an (overly perturbed) sum of a (dynamically chosen) subset of at least $\delta$ honest users' input vectors and some set of seeds used in noise addition by them, meaning intuitively that the adversary can learn `nothing more' than these two pieces of information.
Note that we allow \texttt{SIM} to make a single query to an ideal functionality \texttt{Ideal} that is defined with appropriately chosen $\delta$ as follows

\begin{align*}
	&\textrm{\texttt{Ideal}}^\delta_{\bm{\Delta}_{U \setminus C}, g_{U \setminus C}} (L) = \begin{cases}
		\sum_{u \in L} \tilde{\bm{\Delta}}_u, \{g_{u, k}\}_{u \in L,  \lvert U \setminus L \rvert+ 1 - \lvert C \cap U \rvert \leq k \leq T},  \\
		\hfill \textrm{if \ } L \subseteq (U \setminus C) \textrm{\ and \ } \lvert L \rvert \geq \delta,\\
		\perp, \hfill \textrm{otherwise},
	\end{cases}
\end{align*} where $L$ can be chosen dynamically by \texttt{SIM} at run time.

\begin{lemma}
	There exists a PPT simulator \texttt{SIM} such that for all $\eta$, $t$, $U$, $T$, $C \subseteq U \cup \{S\}$, $\bm{\Delta}_{U \setminus C}$ and $g_{U \setminus C}$. If $2t > \lvert U \rvert + \lvert C \cap U \lvert$
	, the output of \texttt{SIM} is computationally indistinguishable from the output of \texttt{REAL}$^{\eta, t, U, T}_{\pi, C}$:

\begin{equation*}
	\textrm{\texttt{REAL}}^{\eta, t, U,T}_{\pi, C}(M_C, \bm{\Delta}_{U \setminus C}, g_{U \setminus C}) \approx_c \textrm{\texttt{SIM}}_C^{\eta, t, U, T, \textrm{\texttt{Ideal}}^\delta_{\bm{\Delta}_{U \setminus C}, g_{U \setminus C}}}(M_C),
\end{equation*} where $\delta = t - \lvert C \cap U \lvert$.
\label{lemma:security}
\end{lemma}

\begin{proof}
We prove the theorem by a standard hybrid argument.
We will define a simulator \texttt{SIM} through a series of (polynomially many) subsequent modifications to the real execution \texttt{REAL}, so that the views of $M_C$ in any two subsequent executions are computationally indistinguishable.

\begin{enumerate}[label=Hyb$_{\arabic*}$, leftmargin=2.5\parindent]
	\item This random variable is distributed exactly as the view of $M_C$ is \texttt{REAL}, the joint view of the parties $C$ in a real execution of the protocol.
	\item In this hybrid, \texttt{SIM} has access to all the inputs of honest parties $\bm{\Delta}_{U \setminus C}$, and $g_{U \setminus C}$, and runs the full protocol with $M_C$, which includes simulating the random oracle `on the fly' (using a dynamically generated table), the PKI, and the rest of the \textbf{Setup} phase. The view of the adversary here is the same as the previous one.
	\item In this hybrid, \texttt{SIM} replaces the run of \texttt{SecAgg} with the ideal simulation of \texttt{SecAgg} as specified in the original paper.\footnote{See the proof of Theorem A.2 in the full paper of SecAgg~\cite{bonawitz2017practical2}.} As established in the proof there, this hybrid is computationally indistinguishable from the previous one (on condition that $2t > \lvert U \rvert + \lvert C \cap U \rvert$), and differs from it only in polynomially many modifications as follows:
	\begin{itemize}
		\item \texttt{SIM} aborts if $M_C$ provides any of the honest parties $u$ (in \textbf{AdvertiseKey}) with a correct signature w.r.t. an honest $v$'s public key, on $(c_v^{PK} \Vert s_v^{PK})$ different from those sent by $v$.
		\item For any pair of honest users $u$, $v$, \texttt{SIM} encrypts the messages among them (in \textbf{ShareKeys}, before being given to $M_C$) and decrypts them (in \textbf{Unmasking}, after $M_C$ has delivered them) using a uniformly random key (as opposed to the one obtained through the key agreement $\textrm{\textbf{KA.agree}}$ $(c_u^{SK}, c_v^{PK})$).
		\item \texttt{SIM} additionally aborts if $M_C$ succeeds to deliver, in \textbf{ShareKeys}, a message to an honest client $u$ on behalf of another honest client $v$, such that i) the message is different from the message \texttt{SIM} had given $M_C$ in \textbf{ShareKeys}, and ii) the message does not cause the decryption algorithm to fail.
		\item \texttt{SIM} substitutes all the encrypted shares sent between parties of honest users with encryptions of $0$. (It still returns the `real' shares in Round \textbf{Unmasking} as it did before).
		\item \texttt{SIM} aborts if $M_C$ provides any of the honest parties (in \textbf{ConsistencyCheck}) with a signature on a set which correctly verifies w.r.t. the public key of an honest party, but the honest client never produced a signature on that set.
		\item \texttt{SIM} aborts if $M_C$ queries the random oracle/\textbf{PRG} on input $b_u$ for some honest user $u$ (i.e., the value sampled by \texttt{SIM} on behalf of $u$ in \textbf{ShareKeys}) either i) before the adversary received the responses from the honest users in \textbf{Unmasking} or ii) after such responses have been received, but $u \notin Q$.\footnote{We define $Q \subseteq U$ as the only set such that there exists an honest user which received the set Q in \textbf{ConsistencyCheck} and later received at least $t$ valid signatures on it in \textbf{Unmasking}; or empty set if no such $Q$ exists.}
		\item \texttt{SIM} aborts if $M_C$ queries the random oracle/\textbf{PRG} on input $s_{u, v}$ for some honest users $u$, $v$ either i) before the adversary received the responses from the honest users in \textbf{Unmasking} or ii) after such responses have been received, but where $u$, $v \in Q$.
		\item \texttt{SIM} programs the random oracle to set $\textrm{\textbf{PRG}}(b_u)$ for all $u \in Q \setminus C$ to
		\begin{equation*}
			\textrm{\textbf{PRG}}(b_u) \leftarrow \bm{y}_u - \bm{w}_u - \sum_{v \in F_u \setminus Q} \textrm{\textbf{PRG}}(s_{u, v})
		\end{equation*} where $\{\bm{w}_u\}_{u \in Q \setminus C}$ are chosen uniformly at random, subject to $\sum_{u \in Q \setminus C} \bm{w}_u = \sum_{u \in Q \setminus C} \tilde{\bm{\Delta}}_u$ which is obtained by \texttt{SIM} via querying the functionality \texttt{IDEAL} for the set $Q \setminus C$ in \textbf{Unmasking}, $v \in F_u$ iff $M_C$ delivered a ciphertext to $u$ from $v$ in \textbf{ShareKeys}. For all $u \notin Q \setminus C$, \texttt{SIM} sets $\textrm{\textbf{PRG}}(b_u)$ arbitrarily.
	\end{itemize}
	\textit{Note}: This hybrid does not make use of the honest parties' $\bm{\Delta}_{U \setminus C}$, which implies the security of SecAgg against malicious adversaries.
	\item This hybrid is defined exactly as the previous one, except that \texttt{SIM} additionally aborts if $M_C$ queries the random oracle for noise on input $g_{u, k}$ for some honest user $u$ (i.e., the value sampled by \texttt{SIM} on behalf of $u$ in \textbf{Setup}) and either 
	i) before the adversary received the responses from the honest users in \textbf{Unmasking}, or
	ii) after such responses have been received, but $u \notin Q$, or
	iii) after such responses have been received and $u \in Q$, but $k < \lvert U \setminus Q \rvert + 1$.
	
	This hybrid is indistinguishable from the previous one because \texttt{SIM} will abort due to this new condition only if $M_C$ is able to guess one of the $g_{u, k}$, which can only happen with negligible probability (as they are chosen from the exponentially large domain $\mathbb{F}$).
	To see why the view of $M_C$ does not depend on $g_{u, k}$, we can analyze which of the view's components depend on any of those $g_{u, k}$.
	In case i), $M_C$ only receives from \texttt{SIM} at most $\lvert C \cap U \rvert$ shares of $g_{u, k}$ (sent by $u$ in \textbf{ShareKeys}, one for each of the corrupt clients).
	However, since $\lvert C \cap U \rvert < t$, the distribution of any such shares is independent of $g_{u, k}$ (because of the properties of secret sharing).
	In case ii), if $u \notin Q$, the honest user $u$ would abort in \textbf{ConsistencyCheck} and do not join \textbf{Unmasking}.
	No other honest user would send to the server any share of $g_{u, k}$ in \textbf{ExcessiveNoiseRemoval}, either. Thus \texttt{SIM} does not send have to send any to $M_C$.
	Even in case iii), the view of $M_C$ is still independent of $g_{u, k}$ for $k < \lvert U \setminus Q \rvert + 1$: neither would the honest user $u$ send such $g_{u, k}$ to the server in \textbf{Unmasking}, nor would any honest user send the shares of any of them to the server in \textbf{ExcessiveNoiseRemoval}.
	\item This hybrid is defined exactly as the previous one, with the only difference being that \texttt{SIM} now does not receive the inputs of the honest parties, but instead, uses the $\{g_{u, k}\}_{u \in Q \setminus C, \lvert U \setminus (Q \setminus C)) \rvert + 1 - \lvert  \cap U \rvert \leq k \leq T}$ that are already queried from \texttt{Ideal} in \textbf{Unmasking} (see Hyb$_3$) to send  to or share with $M_C$.\footnote{Note that $\lvert U \setminus (Q \setminus C)) \rvert + 1 - \lvert C \cap U \rvert$ (and thus, \texttt{Ideal}) is well-defined as i) it is always no smaller than $1$ given that $\lvert Q \setminus C \rvert \leq \lvert U \rvert - \lvert C \cap U \rvert$, and (ii) it can be larger than $T$, which makes the received $\{\bm{n}_{u, k}\}$ an empty set.}
	More in detail, for all $u \in Q \setminus C$, \texttt{SIM} substitutes its $g_{u, k}$ for $1 \leq k \leq  \lvert U \setminus Q \rvert$ with the corresponding one queried from \texttt{Ideal}, while setting those for $\lvert U \setminus Q \rvert + 1 \leq k \leq T$ with uniformly random values, during both \textbf{ShareKeys} and \textbf{Unmasking}.
	
	Note that this is always feasible for \texttt{SIM} as i) \texttt{Ideal} will not abort (by construction $\lvert Q \rvert \leq t$ and thus $\lvert Q \setminus C \lvert \leq t - \lvert C \cap U \rvert$), and ii) $\lvert U \setminus (Q \setminus C) \rvert + 1 - \lvert C \cap U \rvert \leq \lvert U \setminus Q \rvert + 1$ always hold.
	It is easy to see that this change does not modify the view of the adversary, and therefore it is perfectly indistinguishable from the previous one. Moreover, this hybrid does not make use of the honest party's $g_{U \setminus C}$ (and it does not use $\bm{\Delta}_{U \setminus C}$, either, from Hyb$_3$), and this concludes the proof.
\end{enumerate}

\end{proof}

Given Lemma~\ref{lemma:security}, we now proceed to the following theorem which provides an end-to-end malicious privacy guarantee of \texttt{XNoise} integrated with SecAgg.

\begin{reptheorem}{thm:security}[Privacy against Malicious Adversaries]
	For all $\eta$, $t$, $U$, $T$, $T_C$, $C \subseteq U \cup \{S\}$, $\bm{\Delta}_{U \setminus C}$ and $g_{U \setminus C}$. If $2t > \lvert U \rvert + \lvert C \cap U \lvert$ and $T_C \geq \lvert C \cap U \rvert$, probabilistic polynomial-time (PPT) adversary, given its view of an execution of $\pi$, cannot recover an aggregated update perturbed with noise less than the target level $\sigma_*^2$ with non-negligible probability.
\end{reptheorem}

\begin{proof}
From Lemma~\ref{lemma:security}, we see that from a real execution of $\pi$, no adversary knows anything more than a perturbed aggregated update, $\sum_{u \in L} \tilde{\bm{\Delta}}_u$, and a set of seeds used in the noise addition, $\{g_{u, k}\}_{u \in L, \lvert U \setminus L \rvert + 1 - \lvert C \cap U \rvert \leq k \leq T}$ , for a dynamically chosen subset of clients $L$ (where $\lvert L \rvert \geq t- \lvert C \cap U \rvert$) with more than negligible probability.

Next, the adversary is able to remove the noise $\bm{n}_{u, k}$ from the aggregated update w.r.t. each of the observed seed $g_{u, k}$.
The aggregated perturbed with the least noise, from the view of the adversary, is thus

\begin{align*}
	&\sum_{u \in L} \tilde{\bm{\Delta}}_u - \sum_{u \in L} \sum_{k = \lvert U \setminus L \rvert + 1 - \lvert C \cap U \rvert}^{T} \bm{n}_{u, k} \\
	=& \sum_{u \in L} \bm{\Delta}_u + \sum_{u \in L} \sum_{k=0}^T \bm{n}_{u, k} - \sum_{u \in L} \sum_{k = \lvert U \setminus L \rvert + 1 - \lvert C \cap U \rvert}^{T} \bm{n}_{u, k} \\
	=&  \sum_{u \in L} \bm{\Delta}_u + \sum_{u \in L} \sum_{k=0}^{ \lvert U \setminus L \rvert - \lvert C \cap U \rvert} \bm{n}_{u, k}.
\end{align*}

Given that the noise distribution $\chi$ is assumed to be closed under summation (\cref{sec:enforcement}), and that $\bm{n}_{u, 0} \sim \chi \left(\frac{\sigma_*^2}{\lvert U \rvert} \cdot \frac{t}{t - T_C} \right)$ and $\bm{n}_{u, k} \sim \chi \left(\frac{\sigma_*^2}{(\lvert U \rvert - k + 1) (\lvert U \rvert - k)} \cdot \frac{t}{t - T_C} \right)$ for honest users as specified in $\pi$ (\cref{sec:enforcement_consolidation}), the noise level of the above least-perturbed aggregated update is at

{\small

\begin{align*}
& \lvert L \rvert \left(\frac{\sigma_*^2}{\lvert U \rvert} \cdot \frac{t}{t - T_C} + \sum_{k=1}^{ \lvert U \setminus L \rvert - \lvert C \cap U \rvert} \frac{\sigma_*^2}{(\lvert U \rvert - k + 1) (\lvert U \rvert - k)} \cdot \frac{t}{t - T_C} \right) \\
=& \frac{t\lvert L \rvert\sigma_*^2}{t - T_C} \left( \frac{1}{\lvert U \rvert} + \sum_{k=1}^{ \lvert U \setminus L \rvert - \lvert C \cap U \rvert} \frac{1}{(\lvert U \rvert - k + 1) (\lvert U \rvert - k)} \right) \\
=& \frac{t\lvert L \rvert\sigma_*^2}{t - T_C} \left(\frac{1}{\lvert U \rvert} + \frac{1}{\lvert U \rvert - (\lvert U \setminus L \rvert - \lvert C \cap U \rvert)} - \frac{1}{\lvert U \rvert}  \right) \\
=& \frac{t\lvert L \rvert\sigma_*^2}{t - T_C} \cdot \frac{1}{\lvert L \rvert + \lvert C \cap U \rvert} \geq  \frac{t\sigma_*^2}{t - T_C} \cdot \frac{t - \lvert C \cap U \rvert}{t - \lvert C \cap U \vert + \lvert C \cap U \rvert} \\
=& \frac{t - \lvert C \cap U \rvert}{t - T_c} \sigma_*^2 \geq \sigma_*^2.
\end{align*}

}%

\end{proof}
\section{Optimizing Number of Chunks for Pipelining}~\label{sec:appendix_pipeline}

\begin{table*}[t]
	\centering
	\footnotesize
	\caption{The programming interface provided for developers to customize their own distributed DP algorithms and applications.}
	
	\begin{tabular}{ m{7em} | m{10em} | m{41em}}
		\hline
		\textbf{Category} & \textbf{Base Class} & \textbf{Customization Instruction} \\
		\hline
		\multirow{7}{*}{Distributed DP} & \texttt{ProtocolServer} & Overwrite \texttt{set\_graph\_dict()} to specify the distributed DP workflow for pipeline plan generation. Create one method for coordinating each operation, e.g., \texttt{encode\_data()} to instruct DP encoding. \\
		\cline{2-3}
		& \texttt{ProtocolClient} & Overwrite \texttt{set\_routine()} to specify the handler for each server's request. Create one method for processing each server's request, e.g., \texttt{encode\_data()} to encode local input on request. \\
		\cline{2-3}
		& \texttt{DPHandler} & Overwrite \texttt{init\_params()} to specify how to initialize DP parameters, \texttt{encode\_data()} and \texttt{decode\_data()} to how to perform DP encoding and decoding given a chunk of input, respectively. \\ \cline{2-3}
		& \texttt{AEHandler}, \texttt{KAHandler}, \texttt{PGHandler}, \texttt{SSHandler} & Overwrite the respective functionality for necessary security primitives: authenticated encryption, key agreement, pseudorandom generator, and secret sharing. \\
		\hline
		\multirow{3}{*}{Application} & \texttt{AppServer} & Overwrite \texttt{use\_output()} to specify how the server uses the output of distributed DP. \\
		\cline{2-3}
		& \texttt{AppClient} & Overwrite \texttt{prepare\_data()} and \texttt{use\_output()} to specify how a client prepares the input and consumes the output of distributed DP, respectively.  \\
		\hline
	\end{tabular}
	\label{tab:api}
\end{table*}

Suppose that $m$ is the number of chunks specified by the pipeline execution plan, $\tau_{s}$ is the predicted latency of executing a stage $s$ for any data chunk (note that each chunk has an equal size). 
Moreover, denote by $a$ the total number of stages in the distributed DP workflow, $b_{s, c}$, $f_{s, c}$ the beginning and finishing time of stage $s \in [a]$ for chunk $c \in [m]$, respectively.
To find the optimal number of chunks, $m^{*}$, \system{} solves the following optimization problem which minimizes the end-to-end latency $\tau$, which equates to the finishing time of the last stage for the last chunk, i.e., $f_{a, m}$:

\begin{align}
    m^{*} &= \argmin_{m \in \mathbb{N}_+} f_{a, m}, \nonumber \\
    \textrm{s.t.} \quad  
    f_{s, c} &= b_{s, c} + \tau_{s} \nonumber \\
    b_{s, c} &= \max \{o_{s, c}, r_{s, c}\}, \nonumber \\
    o_{s, c} &= \begin{cases}
        0, & \text{if } s = 0, \\
        f_{s-1, c}, & \text{otherwise},
    \end{cases} \label{eq:basic} \\
    r_{s, c} &= \begin{cases}
        0, & \text{if } s = 0 \text{ and } c = 0, \\
        f_{q, m} \text{ or } \perp, & \text{if } s \neq 0 \text{ and } c = 0, \\
        f_{s, c-1}, & \text{otherwise},\label{eq:exclusive}
    \end{cases}
\end{align}
where $q = \max_{i < s}\{i \mid r_i = r_s \}$ where $r_s$ is the dominant resource of stage $s$.\footnote{We abort the second case in Constraint~\eqref{eq:exclusive} if such $q$ does not exist.}
Note that constraint~\eqref{eq:basic} is enforced as each chunk needs to sequentially go through the execution of all stages. 
Constraint~\eqref{eq:exclusive} is enforced for achieving another two principles in the meantime: 1) each resource can be allocated at most one chunk to execute a stage at any time, and 2) allocating resource $r$ to execute stage $s$ for a chunk $c$ will be suspended, if there exists another chunk $c'$ which has not finished its execution of some previous stage $q < s$ that also uses the resource $r$. 
This optimization problem can be solved by enumeration within a small range (e.g., $m \in [20]$) provided that all $\tau_{s}$'s have been profiled (\cref{sec:pipeline_determining}).
\section{Generic Programming Interface}
~\label{sec:appendix_programming}

\PHB{Support for Diverse Distributed DP Algorithms.}
As outlined in Section~\ref{sec:implementation}, \system{} offers developers the capability to implement customized distributed DP algorithms using the provided generic programming interface.
To facilitate this, Table~\ref{tab:api} presents the base classes that can be tailored to specific requirements.

The first base class, \texttt{ProtocolServer}, serves as a foundation for implementing various \emph{server-side workflows} (as exemplified in Table~\ref{tab:pipeline_stage}).
To streamline the development of computation logic, \system{} supplies established communication primitives based on \texttt{Socket.IO}~\cite{python-socketio}), which developers can employ to handle server-client interactions effectively.
To specify the execution order and resource dependencies of different operations, \system{} further provides the \texttt{set\_graph\_dict()} method, allowing developers to annotate and aid \system{} in planning an optimized pipeline acceleration (see~\cref{sec:pipeline_determining}).

For implementing the \emph{client-side workflow}, developers can utilize the \texttt{ProtocolClient} base class, which can be customized to meet specific requirements.
Additionally, \system{} provides the \texttt{set\_routine()} method, enabling developers to specify which part of the client workflow is triggered by a specific server request.

In addition to facilitating high-level workflow construction, \system{} also equips developers with generic base classes to implement their own \emph{privacy and security primitives}. These include differential privacy mechanisms (\texttt{DPHandler}), authenticated encryption schemes (\texttt{AEHandler}), key agreement protocols (\texttt{KAHandler}), pseudorandom number generators (\texttt{PGHandler}), and secret sharing algorithms (\texttt{SSHandler}).
This empowers developers to easily integrate and experiment with diverse privacy and security building blocks.

\PHM{Support for Diverse Applications.}
Table~\ref{tab:api} also demonstrates that \system{} offers developers the flexibility to extend the aggregation framework to power various privacy-sensitive applications beyond FL. This versatility allows for seamless integration into different use cases.
To achieve this, developers can leverage the \texttt{AppServer} class by overriding the \texttt{use\_output()} method, which specifies how the aggregated result is utilized by the server.
Likewise, developers can instantiate their own \texttt{AppClient} and customize its behavior by overriding two methods. Firstly, the \texttt{prepare\_data()} method can be modified to define how a client's data is generated. Secondly, the \texttt{use\_output()} method can be configured to determine how the aggregated result is consumed by the client.
We have made \system{} open-source and warmly welcome contributions from the research community to further advance privacy-preserving data analytics research.
\section{Artifact Appendix}
~\label{sec:appendix_artifact}

\subsection{Abstract}
We have made the artifact available in our GitHub repository's main branch\footnote{At \url{https://github.com/SamuelGong/Dordis}.} for reproducing the experimental results presented in the paper.
Additionally, for long-term accessibility, we have also uploaded the artifact to a Zenodo repository.\footnote{At \url{https://doi.org/10.5281/zenodo.10023704}.}
The artifact comprises the source code of the \system{}, along with configuration files and Python scripts necessary to execute the experiments described in the paper.
Further instructions on using the artifact can be found in the subsequent sections and in the README file of the repository.

\subsection{Description \& Requirements}

\subsubsection{How to access}
The artifact is available at the above-mentioned repositories.

\subsubsection{Hardware dependencies}
Paper experiments were done in the following two modes:

\begin{itemize}
	\item \emph{Simulation}:
	When assessing the efficacy of noise enforcement (\cref{sec:evaluation_privacy}), each experiment can be effectively simulated using a single node equipped with multiple GPUs for acceleration.
	This is because the metrics of interest, namely privacy budget consumption and final model accuracy, remain unaffected by system speed or network bandwidth.
	As a reference, our evaluation employed a machine with 8 NVIDIA Geforce RTX 2080 Ti GPUs.
	\item \emph{Cluster Deployment}:
	During the evaluation of noise enforcement efficiency (\cref{sec:evaluation_efficiency_noise}) and pipeline acceleration (\cref{sec:evaluation_efficiency_pipeline}), the respective experiments necessitate an EC2 cluster setup.
	In this configuration, an \texttt{r5.4xlarge} instance serves as the server, while each client node is equipped with a \texttt{c5.xlarge} instance.
	This setup aims to replicate the computational capabilities of mobile devices.
	Additionally, the artifact includes scripts for simulating network heterogeneity by limiting the bandwidth of the client instances.
\end{itemize}

\subsubsection{Software dependencies}
The artifact requires an Anaconda environment with Python 3.
While the specific Python packages used are listed in the `requirement.txt' file of the GitHub repository, one can easily install them using provided scripts.

\subsubsection{Benchmarks} 
As mentioned in~\cref{sec:evaluation_methodology}, the provided artifact utilizes publicly available machine learning datasets, namely CIFAR-10~\cite{krizhevsky2009learning}, FEMNIST~\cite{caldas2018leaf}, and Reddit~\cite{reddit}.
Additionally, a variety of models are employed, including ResNet-18~\cite{he2016deep}, a customized CNN~\cite{kairouz2021distributed, agarwal2021skellam}, VGG-19~\cite{simonyan2014very}, and Albert~\cite{lan2020albert}.
The artifact includes automated scripts that can be used to download and preprocess all the datasets.

\subsection{Set-up}
Installation and configuration steps required to prepare the
artifact environment are described in the `Simulation' and `Cluster Deployment' section of the repository README file.

\subsection{Evaluation workflow}

\subsubsection{Major Claims}

As mentioned in~\cref{sec:evaluation}, the major experimental claims made
in the paper are:

\begin{itemize}
    \item C1: Our noise enforcement scheme, \texttt{XNoise}, guarantees the consistent achievement of the desired privacy level, even in the presence of client dropout, while maintaining the model utility.
	This claim is supported by the simulation experiment (E1) outlined in~\cref{sec:evaluation_privacy}, with the corresponding results presented in Figure~\ref{fig:evaluation_privacy}, Table~\ref{tab:utility}, and Figure~\ref{fig:evaluation_round_acc}.
    \item C2: The \texttt{XNoise} scheme introduces acceptable runtime overhead, with the network overhead being scalable with respect to the model's expanding size. S
    This can be proven by the cluster experiment (E2) described in~\ref{sec:evaluation_efficiency_noise} whose results are reported in Figure~\ref{fig:evaluation_efficiency} and Table~\ref{tab:network}.
    \item C3: The pipeline-parallel aggregation design employed by \system{} significantly boosts training speed, leading to a remarkable improvement of up to 2.4$\times$ in the training round time.
    This finding is supported by the cluster experiment (E3) discussed in~\cref{sec:evaluation_efficiency_pipeline}, and the corresponding results are depicted in Figure~\ref{fig:evaluation_efficiency}.
\end{itemize}

\subsubsection{Experiments}
The artifact includes a dedicated section in the repository's README file titled `Reproducing Experimental Results'. This section provides a comprehensive guide on conducting the necessary experiments to replicate the main claims mentioned earlier.
For every experiment, we meticulously document all the essential setup requirements, estimated costs and timeframes, a concise explanation of the procedure, the anticipated outcomes, and an explicit link to one of the aforementioned three claims.

\end{document}